\documentclass[sn-basic,Numbered]{sn-jnl}
\RequirePackage{amsthm}

\usepackage{graphicx}%
\usepackage{multirow}%
\usepackage{amsmath,amssymb,amsfonts}%
\usepackage{amsthm}%
\usepackage{mathrsfs}%
\usepackage[title]{appendix}%
\usepackage{xcolor}%
\usepackage{textcomp}%
\usepackage{manyfoot}%
\usepackage{booktabs}%
\usepackage{algorithm}%
\usepackage{algorithmicx}%
\usepackage{algpseudocode}%
\usepackage{listings}%
\usepackage{subcaption}%

\theoremstyle{thmstyleone}%
\newtheorem{theorem}{Theorem}
\newtheorem{proposition}{Proposition}
\newtheorem{lemma}{Lemma}
\newtheorem{example}{Example}

\theoremstyle{thmstylethree}%

\usepackage{comment}
\usepackage{mathtools}
\usepackage{enumitem}
\usepackage{hyperref}
\usepackage{wrapfig}
\usepackage{tikz}
\usetikzlibrary{arrows,automata,shapes,positioning,patterns,arrows.meta}
\usepackage {mathpartir}
\usepackage{hhline}
\usepackage{content/ltl}
\usepackage{xspace}
\DeclareMathOperator{\cf}{\LTLsquare\!\!\rightarrow}

\DeclareMathOperator{\mcf}{\LTLdiamond\hspace{-3.5pt}\rightarrow}

\DeclareMathOperator{\ucf}{\LTLhalfsquare\!\!\rightarrow}

\DeclareMathOperator{\emcf}{\LTLhalfdiamond\hspace{-3.5pt}\rightarrow}

\newcommand{\ltl}{LTL\xspace}
\newcommand{\kltl}{KLTL\xspace}
\newcommand{\yltl}{YLTL\xspace}

\newcommand{\fo}{FO[$<$]\xspace}
\newcommand{\foe}{FO[$<$,\textit{E}]\xspace}
\newcommand{\AP}{\mathit{AP}}

\newcommand{\tofo}{\mathit{fo}}

\newcommand{\ags}{A}

\newcommand{\true}[0]{\mathit{true}}
\newcommand{\false}[0]{\mathit{false}}

\newcommand{\ldot}{\mathpunct{.}}

\newcommand{\A}{\mathcal{A}}
\newcommand{\U}{\LTLuntil}

\newcommand{\G}{\LTLglobally}

\newcommand{\K}{\mathtt{K}}

\newcommand{\kripke}{\mathcal{K}}
\newcommand{\ekripke}{\mathcal{E}}
\newcommand{\mods}{\mathit{Mod}}
\newcommand{\suc}{\mathit{succ}}
\newcommand{\mini}{\mathit{min}}

\renewcommand{\models}{\vDash}
\newcommand{\nmodels}{\nvDash}
\newcommand{\ap}{\text{AP}}

\newcommand{\donotshow}[1]{}

\begin{document}

\title[Explainability Requirements as Hyperproperties]{Explainability Requirements as Hyperproperties\vspace{0.5em}}

\author{\fnm{Bernd} \sur{Finkbeiner}}

\author{\fnm{Julian} \sur{Siber}}

\affil{\orgname{CISPA Helmholtz Center for Information Security}, \orgaddress{\street{Stuhlsatzenhaus 5}, \city{Saarbrücken}, \postcode{66123}, \state{Saarland}, \country{Germany}}}

\abstract{
  Explainability is emerging as a key requirement for autonomous systems. While many works have focused on what constitutes a valid explanation, few have considered formalizing explainability as a system property. In this work, we approach this problem from the perspective of hyperproperties. We start with a combination of three prominent flavors of modal logic and show how they can be used for specifying and verifying counterfactual explainability in multi-agent systems: With Lewis' counterfactuals, linear-time temporal logic, and a knowledge modality, we can reason about whether agents know \emph{why} a specific observation occurs, i.e., whether that observation is \emph{explainable} to them. We use this logic to formalize multiple notions of explainability on the system level. We then show how this logic can be embedded into a hyperlogic. Notably, from this analysis we conclude that the model-checking problem of our logic is decidable, which paves the way for the automated verification of explainability requirements.}

\maketitle

\section{Introduction}\label{sec:intro}
The increase in system complexity and opaqueness perceived in recent years has been answered by a plethora of techniques aimed at providing some sort of explanation for observed system behavior~\cite{AlmagorL20,BrandaoMMLC22,Wachter18,RosenfeldR19}. While this demonstrates a need for systems to be explainable, there is no formal theory to specify different notions of explainability and to algorithmically verify them. In this paper, we make the claim that hyperproperties, and their respective logics, are an excellent basis for such a formal theory of explainability. We start from previous theories for individual instances of \emph{explanations}~\cite{HalpernP05a,HalpernP05b}, which combine counterfactual and epistemic reasoning. Besides extending them to system specifications, we add temporal reasoning to specify explainability on the possibly infinite executions of multi-agent systems. We use modal operators for these three reasoning dimensions to express explainability requirements such as:
$$
	\G \Big( \lnot \mathit{offer} \rightarrow  \big( \bigvee_{\alpha,\beta \in \mathit{Att}(a)} \K_a \left((\alpha \land \beta) \mcf_a \mathit{offer}\right)\big)\Big) \enspace ,
$$
which we simply term \emph{Internal Counterfactual Explainability} (ICE). Interpreted in a hiring system where some agent~$a$ applies to get a job offer, ICE states that, whenever agent~$a$ does not get the offer (i.e., atomic proposition $\mathit{offer}$ does not hold), they know that if they had applied with some (other) attribute values $\alpha, \beta \in\mathit{Att}(a)$, they would have gotten the offer. We call this notion \emph{internal} because it depends only on actions performed by agent $a$ themselves. The formula for ICE uses operators from all three modal logics that we fuse together: It uses the temporal operator $\G$ to specify that the requirement holds at every time point and it uses the knowledge operator $\K_a$ to express that agent~$a$ has knowledge about some counterfactual dependency expressed with the counterfactual operator $\mcf_a$.  Later on, we will use this logic to formalize other aspects of explainability, such as \emph{weak}, \emph{external}, and \emph{general} explainability. We will also see how these notions discriminate between -- intuitively -- explainable and unexplainable systems in Section~\ref{sec:motivation}. This appeal to intuition is without alternative: There is no universally correct definition of explainability~\cite{KohlBLOSB19} and much depends on the context and the agents involved. The strength of our modal-logic approach is exactly that it provides a flexible specification language that can be applied to varying contexts and definitions, while retaining a general model-checking algorithm.

Double-fusions of the three modal logics we consider have been studied extensively: Epistemic temporal logic has been used in security for information-flow control~\cite{HalpernO08,BalliuDG11}, counterfactual temporal logic for expressing causal dependencies in reactive systems~\cite{FinkbeinerS23,CoenenFFHMS22}, and counterfactual epistemic logics to characterize notions of rationality in game theory ~\cite{Stalnaker06,Sandu21}. Our work brings these diverse frameworks together based on the viewpoint that explainability is an \emph{intended flow of information about counterfactual dependencies}. This interpretation stands in the tradition of a long line of works on \emph{individual} causal explanations~\cite{Lewis86,HalpernP05b,Beckers22,GoyalWEBPL19}. With this paper, we shift the focus away from the question of what constitutes a valid individual explanation toward analyzing the abstract epistemic properties that the global system needs to fulfill such that an explanation is available to an agent whenever it is needed. In short, we do not analyze explanations, but \emph{explainability}.

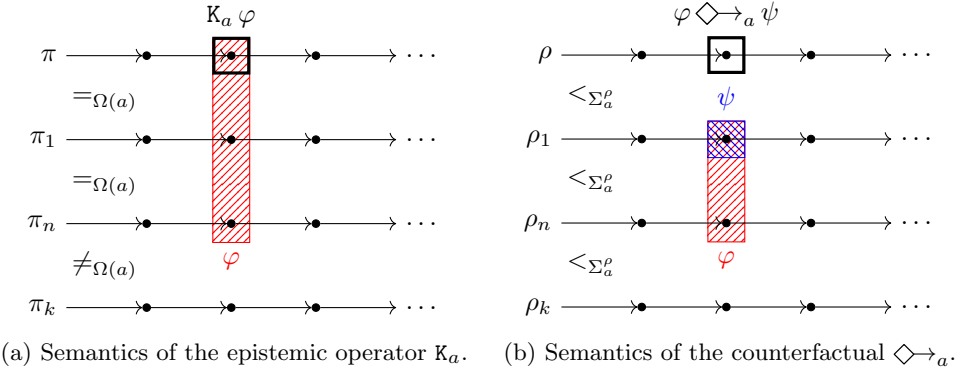
\begin{figure}
    \centering
    \begin{subfigure}{.49\textwidth}
        \centering
    \begin{tikzpicture}[auto,
			node distance=1 and 1,every state/.style={minimum size=2pt,inner sep=1pt,fill},
            square/.style={regular polygon,minimum size=0.65cm,regular polygon sides=4}]
        
        \node[draw,rectangle,
                minimum width =0.485cm, 
                minimum height = 2.7cm,pattern=north east lines, pattern color=red,draw=red](s) at (1.113,1.1) {};
    
        \node[state](2){};
        \node[left = of 2](1){$\pi_n$};
        \node[state,right = of 2](3){};
        \node[state,right = of 3](4){};
        \node[right = of 4](5){\dots};
        \path[->,draw] (1) edge (2)
        (2) edge (3)
        (3) edge (4)
        (4) edge (5);
        
        \node[state,above = of 2](2a){};
        \node[left = of 2a](1a){$\pi_1$};
        \node[state,right = of 2a](3a){};
        \node[state,right = of 3a](4a){};
        \node[right = of 4a](5a){\dots};
        \path[->,draw] (1a) edge (2a)
        (2a) edge (3a)
        (3a) edge (4a)
        (4a) edge (5a);

        \node[state,above = of 2a](2b){};
        \node[left  = of  2b](1b){$\pi$};
        \node[state,right = of 2b](3b){};
        \node[state,right = of 3b](4b){};
        \node[right = of 4b](5b){\dots};
        \path[->,draw] (1b) edge (2b)
        (2b) edge (3b)
        (3b) edge (4b)
        (4b) edge (5b);
        
        \node[state,below = of 2](2c){};
        \node[left  = of  2c](1c){$\pi_k$};
        \node[state,right = of 2c](3c){};
        \node[state,right = of 3c](4c){};
        \node[right = of 4c](5c){\dots};
        \path[->,draw] (1c) edge (2c)
        (2c) edge (3c)
        (3c) edge (4c)
        (4c) edge (5c);
        
        \node[above= 0.2 of 3b](a){$\K_a \,\varphi$};
        \node[draw,square,very thick](s) at (3b) {};
        \node[below= 0.2 of 3](a){\textcolor{red}{$\varphi$}};
        \node[above right = 0.02 and -0.05 of 1a]{$=_{\Omega(a)}$};
        \node[above right = 0.02 and -0.05 of 1]{$=_{\Omega(a)}$};
        \node[above right = 0.02 and -0.05 of 1c]{$\neq_{\Omega(a)}$};
    \end{tikzpicture}
    \caption{Semantics of the epistemic operator $\K_a$.}
    \label{subfig:sem1}
    \end{subfigure}
    \begin{subfigure}{.49\textwidth}
        \centering
    \begin{tikzpicture}[auto,
			node distance=1 and 1,every state/.style={minimum size=2pt,inner sep=1pt,fill},
            square/.style={regular polygon,minimum size=0.65cm,regular polygon sides=4}]

        \node[draw,rectangle,
                minimum width =0.485cm, 
                minimum height = 1.6cm,pattern=north east lines, pattern color=red,draw=red](s) at (1.113,0.55) {};
                
        \node[draw,rectangle,
                minimum width =0.485cm, 
                minimum height = 0.485cm,pattern=north west lines, pattern color=blue,draw=blue](s) at (1.113,1.11) {};

        \node[state](2){};
        \node[left = of 2](1){$\rho_n$};
        \node[state,right = of 2](3){};
        \node[state,right = of 3](4){};
        \node[right = of 4](5){\dots};
        \path[->,draw] (1) edge (2)
        (2) edge (3)
        (3) edge (4)
        (4) edge (5);
        
        \node[state,above = of 2](2a){};
        \node[left = of 2a](1a){$\rho_1$};
        \node[state,right = of 2a](3a){};
        \node[state,right = of 3a](4a){};
        \node[right = of 4a](5a){\dots};
        \path[->,draw] (1a) edge (2a)
        (2a) edge (3a)
        (3a) edge (4a)
        (4a) edge (5a);

        \node[state,above = of 2a](2b){};
        \node[left  = of  2b](1b){$\rho$};
        \node[state,right = of 2b](3b){};
        \node[state,right = of 3b](4b){};
        \node[right = of 4b](5b){\dots};
        \path[->,draw] (1b) edge (2b)
        (2b) edge (3b)
        (3b) edge (4b)
        (4b) edge (5b);
        
        \node[state,below = of 2](2c){};
        \node[left  = of  2c](1c){$\rho_k$};
        \node[state,right = of 2c](3c){};
        \node[state,right = of 3c](4c){};
        \node[right = of 4c](5c){\dots};
        \path[->,draw] (1c) edge (2c)
        (2c) edge (3c)
        (3c) edge (4c)
        (4c) edge (5c);
        
        \node[above= 0.2 of 3b](a){$\varphi \mcf_a \psi$};
        \node[draw,square,very thick](s) at (3b) {};
        \node[below= 0.2 of 3](a){\textcolor{red}{$\varphi$}};
        \node[above= 0.2 of 3a](a){\textcolor{blue}{$\psi$}};
        \node[above right = 0.05 and -0.05 of 1a]{$<_{\Sigma_a^\rho}$};
        \node[above right = 0.05 and -0.05 of 1]{$<_{\Sigma_a^\rho}$};
        \node[above right = 0.05 and -0.05 of 1c]{$<_{\Sigma_a^\rho}$};
    \end{tikzpicture}
    \caption{Semantics of the counterfactual $\mcf_a$.}
    \label{subfig:sem2}
    \end{subfigure}
    \caption{Illustrating the semantics of the statements $\K_a \, \varphi$ and $\varphi \mcf_a \psi$ on a set of traces. The statements are evaluated at the second position of the trace $\pi$, which is framed by the black square. The epistemic operator $\K_a$ (cf.~Subfigure~\ref{subfig:sem1}) requires $\varphi$ to hold at the same position on all traces $\pi'$ with an observation equivalent prefix, i.e., where $\pi =_{\Omega(a)} \pi'$ is satisfied. These positions are covered by the area with diagonal lines (colored red). In contrast, the counterfactual $\mcf_a$ (cf.~Subfigure~\ref{subfig:sem2}) requires that the trace \emph{closest} to $\rho$ that satisfies $\varphi$, which is in this case $\rho_1$, also satisfies $\psi$ at the same position. This is marked by the area with crossed lines (colored red and blue). $<_{\Sigma_a^\rho}$ is used in the illustration to denote, e.g., $(\rho_n,\rho_k) \in \Sigma_a^\rho \land (\rho_k,\rho_n) \notin \Sigma_a^\rho$, which means that $\rho_n$ is strictly more similar to $\rho$ than $\rho_k$. }
    \label{fig:sem}
\end{figure}

\paragraph*{Knowledge and Counterfactuals.} 

The logic we consider is an extension of epistemic temporal logic, in particular of Linear Temporal Logic (\ltl) with the knowledge modality $\K$ (\kltl). We extend KLTL with counterfactual conditionals as defined by Lewis~\cite{Lewis73}, which we interpret on paths of a multi-agent system. 

We illustrate the semantics of these two modal operators in Figure~\ref{fig:sem}. The epistemic formula $\K_a \, \varphi$ is satisfied at a given position of a given trace if all traces that are indistinguishable for agent a satisfy $\varphi$ (cf.~Subfigure~\ref{subfig:sem1}). Indistinguishability is defined based on the observation-equivalence $=_{\Omega(a)}$ that compares the prefixes of two traces with respect to the observations of agent $a$. The Lewisian counterfactual $\varphi \mcf_a \psi$, on the other hand, informally has the following meaning: ``If $\varphi$ had been true then $\psi$ might also have been true''. More formally, the counterfactual formula has the following semantics: it holds on a position~$i$ of a given trace if one of the \emph{closest} traces that satisfy $\varphi$ at $i$ also satisfies $\psi$ at $i$. These semantics are illustrated in in Subfigure~\ref{subfig:sem2}, where there is in fact a unique closest trace to $\rho$ that satisfies $\varphi$, which also satisfies $\psi$, such that the counterfactual formula holds. Closeness of traces is modeled through a binary similarity relation $\Sigma_a^\rho$ that defines whether some trace is at least as similar to $\rho$ as another trace. Our approach is parametric for varying agent-specific similarity metrics, such that $\Sigma_a^\rho$ depends on agent $a$. This, in particular, allows to model different internal causal models for different agents. We follow Lewis' formulation of counterfactuals and do not assume that there is a unique closest execution for every antecedent in a counterfactual, which is often termed the limit assumption and endorsed by the competing counterfactual theory of Stalnaker~\cite{Stalnaker81}. It has been noted in previous work that this assumption is easily violated when combining counterfactuals and temporal logics~\cite{FinkbeinerS23}. 

The combination of knowledge and counterfactual operators gives a specification like ICE the following semantics: it holds at a given position if there is a combination of attribute values $\alpha$ and $\beta$ such that on all traces $\rho$ that are indistinguishable for agent $a$, making the minimal changes to the trace such that $\alpha \land \beta$ holds results in a trace where agent $a$ gets the offer. The nature of the minimal changes is defined by the similarity relation $\Sigma_a^\rho$, i.e., the internal causal model of agent $a$. Hence, ICE requires that on any trace agent $a$ knows about some counterfactual explanation $\alpha \land \beta$ for the outcome $\mathit{offer}$ whenever this outcome does not happen.

\paragraph{Expressivity and Model Checking.} For the logic of the combined three modal systems, we construct a translation function that maps formulas to sentences in first-order logic of order with an equal-level predicate (\foe)~\cite{FinkbeinerZ17}. This logic allows to quantify over tuples of traces and positions and hence it is a logic for hyperproperties~\cite{ClarksonS10}. Our translation serves two purposes. On the one hand, it is a first result on the comparative expressiveness of this logic in relation to other hyperlogics, i.e., it places our logic for explainability into the hierarchy of hyperlogics~\cite{CoenenFHH19}. On the other hand, it proves that model-checking formulas in this logic on finite-state multi-agent systems is decidable, and provides an algorithm via the proposed encoding into \foe. As far as we know, this is the first positive decidability result for model-checking of arbitrarily nested temporal and counterfactual operators, as our earlier study relegated counterfactuals to top-level operators~\cite{FinkbeinerS23}. Moreover, this previous work did not include knowledge operators which are necessary for formalizing explainability.

\paragraph*{Contributions.} 
In short, we make the following contributions.
\begin{itemize}
	\item We define a combined logic of counterfactuals, knowledge and temporal modalities on the executions of multi-agent systems. This is an extension of our earlier work that did not consider knowledge operators~\cite{FinkbeinerS23}.
	\item We formalize multiple notions of explainability in this logic, demonstrate practically how they distinguish explainable systems, and theoretically study their entailment relations.
	\item We outline a model-checking algorithm for this logic on multi-agent systems with a finite state-space. This algorithm relies on an encoding into the hyperlogic \foe, which also yields some first insights into the comparative expressiveness of our presented logic.
\end{itemize}

\section{Preliminaries} 
We recall some background on extended Kripke structures as models of multi-agent systems and on temporal, epistemic, and hyper logics as specification languages.

\subsection{Multi-Agent Systems}

We consider \emph{Kripke structures} as the fundamental model of temporal logic. A Kripke structure $\kripke = (S,s_0,\Delta,\AP,\Lambda)$ is a tuple, where $S$ is a set of states, $s_0$ is the initial state, $\Delta: S \mapsto 2^S$ is a transition function such that $\Delta(s) \neq \emptyset$ for all states $s \in S$, $\ap$ is a set of \emph{atomic propositions}, and $\Lambda : S \mapsto 2^\AP$ is a function labeling states with atomic propositions. We call $\kripke$ a \emph{finite} Kripke structure if the set of states $S$ is finite. A \emph{path} $\rho = \rho[0] \rho[1] \ldots \in S^\omega$ of a Kripke structure $K$ is an infinite sequence of states such that the transition function is respected: $\rho[i+1] \in \Delta(\rho[i])$ for all $i \in \mathbb{N}$. The \emph{trace} $\pi =  \pi[0] \pi[1] \ldots \in (2^\AP)^\omega$ on a path $\rho$ is the sequence of corresponding state labels, i.e., we have $\pi[i] = \Lambda(\rho[i])$ for all $i \in \mathbb{N}$. Let $\Pi(\kripke)$ denote the set of traces on initial paths starting in $s_0$, i.e., on $\rho$ such that $\rho[0] = s_0$. For some trace $\pi$, $\pi[0,n] \in S^*$ denotes its prefix of length $n+1$.
We can extend a Kripke structure $\kripke$ with an \emph{observation map} $\Omega : \ags \mapsto 2^\AP$ to reason about the local observations of a set of agents $\ags$. For some agent $a \in \ags$, $\Omega(a)$ describes the set of atomic propositions that are observable to agent $a$. For some trace $\pi$ of $\kripke$, $\Omega_a(\pi) \in (2^\AP)^\omega$ are the partial observations of $a$ along the trace: $\Omega_a(\pi)[i] = \pi[i] \cap \Omega(a)$. We say $\ekripke = (\kripke,\Omega)$ is finite if $\kripke$ is finite. The set of traces of $\ekripke = (\kripke,\Omega)$ is denoted $\Pi(\ekripke) = \Pi(\kripke)$. The set of extended Kripke structures that satisfy some logical formula $\varphi$ is denoted by $\mods(\varphi)$.

\subsection{Epistemic Temporal Logic}\label{subsec:logics}
The basis of our logic is \kltl, which extends Linear Temporal Logic (\ltl)~\cite{Pnueli77} with a knowledge modality~\cite{FaginHMV1995}. The syntax of \kltl is defined by the following grammar:
$$
	\varphi \Coloneqq p \mid \neg \varphi \mid \varphi \lor \varphi \mid \LTLnext \varphi \mid \varphi \U \varphi \mid \K_a \, \varphi \mid \LTLnext^- \varphi \mid\LTLuntil^- \varphi\enspace ,
$$
where $p \in \AP$ is an atomic proposition and $a \in \ags$ is an agent. Additionally, \kltl includes the following derived operators: Boolean constants ($\true$, $\false$) and connectives ($\lor$, $\rightarrow$, $\leftrightarrow$), and the temporal operator `Eventually' ($\LTLeventually \varphi \equiv \true \LTLuntil \varphi$) as well as its dual, `Globally' ($\LTLglobally \varphi \equiv \lnot \LTLeventually \lnot \varphi $). The semantics of a \kltl formula $\varphi$ with respect to an extended Kripke structure $\ekripke = (\kripke,\Omega)$, a trace $\pi \in \Pi(\kripke)$, and a position $i$ is defined by the following satisfaction relation:
\begin{alignat*}{2}
	&\ekripke,\pi,i \models p	&~\text{ iff }~	&p \in \pi[i],\\
	&\ekripke,\pi,i \models \lnot \varphi	&~\text{ iff }~	& \ekripke,\pi,i \nmodels \varphi,\\
	&\ekripke,\pi,i \models \varphi_1 \lor \varphi_2	&~\text{ iff }~	& \ekripke,\pi,i \models \varphi_1 \lor \ekripke,\pi,i \models \varphi_2,\\
	&\ekripke,\pi,i \models \LTLnext \varphi	&~\text{ iff }~	& \ekripke,\pi,i+1 \models \varphi,\\
	&\ekripke,\pi,i \models \varphi_1 \LTLuntil \varphi_2	&~\text{ iff }~	& \exists k \geq i: \ekripke,\pi,k \models \varphi_2 \land \forall i \leq j < k: \ekripke,\pi,j \models \varphi_1,\\
	&\ekripke,\pi,i \models \K_a \, \varphi	&~\text{ iff }~	& \forall \pi' \in \Pi(\kripke): (\Omega_a(\pi)[0,i] = \Omega_a(\pi')[0,i]) \rightarrow \ekripke,\pi',i \models \varphi .
\end{alignat*}
Hence, an agent $a$ has knowledge of some property $\varphi$, expressed through $\K_a(\varphi)$, iff this property holds on all observation-equivalent prefixes of the same length. These semantics of the knowledge modality $\K_a$ correspond to the so-called synchronous perfect recall semantics~\cite{MeydenS99,HalpernMV04}, which means that agents gain knowledge through distinguishing prefixes of different length and based on divergence at any point in the past.
System-level satisfaction is based on a universal application of the trace semantics: $\ekripke = (\kripke,\Omega)$ satisfies $\varphi$, denoted by $\ekripke \models \varphi$, iff for all traces $\pi \in \Pi(\kripke): \ekripke,\pi,0 \models \varphi$. We denote the set of \kltl formulas over some alphabet $\AP$ by $\mathcal{L}_{\text{\kltl}}(\AP)$.

\paragraph*{Past-operators} Since an explanation for some effect is usually found in its past, we use \kltl with temporal past-operators. We define these operators as usual in the literature~\cite{LichtensteinPZ85}. Given an extended Kripke structure $\ekripke = (\kripke,\Omega)$, an initial trace $\pi \in \Pi(\kripke)$, and a position $i$, we define the semantics of the past-operators as follows:
\begin{alignat*}{2}
	&\ekripke,\pi,i \models \LTLnext^- \varphi	&~\text{ iff }~	& i > 0 \land \ekripke,\pi,i-1 \models \varphi,\\
	&\ekripke,\pi,i \models \varphi_1 \LTLuntil^- \varphi_2	&~\text{ iff }~	& \exists k \leq i: \ekripke,\pi,k \models \varphi_2 \, \land \forall i \geq j > k: \ekripke,\pi,j \models \varphi_1 \enspace.
\end{alignat*}
The `Before' modality $\LTLnext^-$ refers to a previous time point, we define it such that it is trivially false at the start of a given trace. The `Since' operator $\LTLuntil^-$ is a mirror image of `Until' ($\LTLuntil$): It requires that $\varphi_2$ was true at some earlier time point $k$, and that $\varphi_1$ holds on all time points in between. We also add the derived past operators `Once' ($\LTLeventually^- \varphi \equiv \true \LTLuntil^- \varphi$), as well as its dual, `Historically' ($\LTLglobally^- \varphi \equiv \lnot \LTLeventually^- \lnot \varphi $).

\section{In-Depth Example}\label{sec:motivation}

We illustrate our approach for specifying explainability at the example of a simplified hiring system consisting of two agents: Applicant and Recruiter. The high-level idea is that, in every round, Recruiter chooses their preferred values for two attributes \emph{job} and \emph{gender}, and Applicant chooses the attribute values with which they apply in that round. Applicant gets an offer in some round if their attributes match Recruiter's preference. The hiring goes on infinitely, such that Applicant effectively models a stream of applicants applying at the company. The difference between the explainable and unexplainable version of the hiring system is that in the former, the preference of Recruiter is observable to Applicant, while it is hidden in the latter, unexplainable hiring system. Crucially, in both versions Applicant and Recruiter fix their attributes \emph{concurrently}, such that in the explainable hiring scenario Applicant only observes Recruiter's preference \emph{after} the decision. This means that the outcomes Applicant can enforce are the same in both scenarios, e.g., in neither scenario Applicant has a strategy that ensures that they will eventually get an offer. What is different, however, is that in the explainable hiring system Applicant gains knowledge on \emph{why exactly} they did not get the offer in some round, while in the unexplainable system Applicant only knows that they should have done \emph{something} differently.

\subsection{Hiring System Model}\label{subsec:hiring}

To develop this hiring example more formally, consider the following Kripke structure $\kripke = (S,s_0,\Delta,\AP,\Lambda)$ underlying both the explainable and unexplainable hiring system. The set of states $S$ is determined by the different values the attributes of Applicant and Recruiter may have:
\begin{align*}
	S = \big\{(a_{\mathit{job}},a_{\mathit{gen}},r_{\mathit{job}},r_{\mathit{gen}}) \mid \; & a_{\mathit{job}},r_{\mathit{job}} \in \{\mathit{accounting},\mathit{sales},\mathit{it}\}\\ & \land \; a_{\mathit{gen}},r_{\mathit{gen}} \in \{\mathit{m},\mathit{f}\} \big\} \cup \{s_0\} \enspace ,
\end{align*}
where $s_0$ is a unique initial state. Since every round is effectively a new, closed hiring process, the transition function $\Delta$ connects every state with itself and every other state, i.e., $\Delta(s) = S$ for all $s \in S$, so that the underlying graph is fully connected. 

The set of atomic propositions resembles the attribute choices of the two agents. For some agent $x$ we define the corresponding attributes with the function $\mathit{Att}^+$:
\begin{align*}
	\mathit{Att}^+(x) &= \{x_v \mid v \in \{\mathit{accounting},\mathit{sales},\mathit{m},\mathit{f}\} \} \enspace,\\
	\AP &= \mathit{Att}^+(a) \cup \mathit{Att}^+(r) \cup \{\mathit{offer}\} \enspace .
\end{align*}

For the specifications we sometimes need both positive and negated atomic propositions for attributes, which is covered by the function $\mathit{Att}$:
\begin{align*}
	\mathit{Att}(x) &= \mathit{Att}^+(x) \cup \{\lnot p \mid p \in \mathit{Att}^+(x) \} \enspace .
\end{align*}

The labeling function $\Lambda$ labels each state with the attributes picked by the agents, and with $\mathit{offer}$ if they are matching. The initial state is labeled with the empty set:
\begin{align*}
	\Lambda(s_0) = &\{\} ,\enspace
	\Lambda((x,y,v,w)) = \{a_x,a_y,r_v,r_w\} \cup \{ \mathit{offer} \mid x = v \land y = w\} \enspace .
\end{align*}

The unexplainable hiring system $\mathcal{U} = (\kripke,\Omega^\mathcal{N})$ differs from the explainable one $\ekripke = (\kripke,\Omega^\ekripke)$  only in the observation map. We have for an agent $x \in \{a,r\}$:
\begin{align*}
	\Omega^\ekripke(x) = \AP = \mathit{Att}^+(a) \cup \mathit{Att}^+(r) \cup \{\mathit{offer}\} ,\enspace
	\Omega^\mathcal{U}(x) = \mathit{Att}^+(x) \cup \{\mathit{offer}\} \enspace .
\end{align*}
Hence, in the explainable hiring system Applicant can observe the preferences of Recruiter (retrospectively), while in the unexplainable hiring system this information is hidden.

As discussed in Section~\ref{sec:intro} with reference to Figure~\ref{fig:sem}, the semantics of counterfactuals are defined with respect to a similarity relation $\Sigma_a$, which -- intuitively speaking -- encodes the minimal changes that are necessary to go from one trace to another, based on the intenral model of agent $a$. We now give a concrete similarity relation for the application scenario. Here, we only define a relation for Applicant as the properties we consider contain only counterfactuals indexed by $a$. A trace $\pi_{1}$ is \emph{at least as similar} to the reference trace $\pi$ at a given time point $i$ as some other trace $\pi_2$ from the perspective of agent $a$, if the following formula holds at $i$, where pairs with the trace variables $\pi,\pi_1,\pi_2$ are used to refer to atomic propositions on a specific trace:
\begin{align*}
	\Sigma(a)(\pi,\pi_1,\pi_2) = \, &\LTLglobally \Big(\!\!\bigwedge_{p \in A} \!\!\!\big((p,\pi) \not\leftrightarrow (p,\pi_1)\big) \rightarrow  \big((p,\pi) \not\leftrightarrow (p,\pi_2)\big)\Big) \, \land \\
	&\LTLglobally^- \Big(\!\!\bigwedge_{p \in A} \!\!\!\big((p,\pi) \not\leftrightarrow (p,\pi_1)\big) \rightarrow  \big((p,\pi) \not\leftrightarrow (p,\pi_2)\big)\Big) \enspace ,
\end{align*}
where $A = \mathit{Att}^+(a) \cup \mathit{Att}^+(r)$.
Detailed semantics of this relational property follow in Section~\ref{subsec:semantics}. In the formula, the `Historically' operator $\LTLglobally^-$ is the past-time version of the `Globally' operator $\LTLglobally$, which imposes a constraint on all previous time points in a trace. Combined with the regular `Globally' operator, the above formula expresses that the specified requirement does not only hold in the future but also in the past. Note that the past-operator's detailed semantics are given in Section~\ref{sec:yltl}. The specified requirement that is invariant in both past and future states that changes between the reference trace $\pi$ and the at least as similar trace $\pi_1$ also have to be present in the less similar trace $\pi_2$. The changes between $\pi$ and $\pi_2$ can be a proper superset of the changes between $\pi$ and $\pi_1$, but it may also be that $\pi_1$ and $\pi_2$ are identical. Such subset-based similarity has been applied in many notions of causality~\cite{HalpernP05a,CoenenFFHMS22,CoenenDFFHHMS22}. We will see examples of tuples of traces that are in the similarity relation in the following. Take note that formulas encoding the similarity relation are KLTL formulas over a modified alphabet, i.e., the alphabet is $\AP \times \Pi$ where $\Pi$ is a set of trace variables. This is because the similarity relation needs to relate three traces with each other: it is a hyperproperty~\cite{ClarksonS10}.

\subsection{Semantics of Explainability}\label{subsec:example_traces}

To see how the explainability requirement specified by ICE (cf.~Section~\ref{sec:intro}) discriminates between these two hiring systems, consider the following infinite trace $\pi \in \kripke$, which is present in both systems:
\begin{align*}
\pi = \{\}\{a_\mathit{it},a_\mathit{f},r_\mathit{sales},r_\mathit{f}\} \{\}^\omega \enspace .
\end{align*}
The $\omega$-superscript indicates that this part of the trace is repeated infinitely often, i.e., in this case the trace ends up looping in the initial state.
Let us now check whether trace $\pi$ satisfies the requirement posed by ICE. Hence, we now check the semantics that we described abstractly in Section~\ref{sec:intro} with respect to Figure~\ref{fig:sem} for this specific trace $\pi$. The ICE requirement states that at all positions where $\mathit{offer}$ does not hold, the knowledge predicate $\K_a \left((\alpha \land \beta) \mcf_a \mathit{offer}\right)$ has to hold for at least one pair of attributes $\alpha,\beta \in \mathit{Att}(a)$. By the semantics of $\K$, this is the case if the counterfactual conditional $(\alpha \land \beta) \mcf_a \mathit{offer}$ holds at this position on \emph{all} traces that are indistinguishable for Applicant (cf.~Subfigure~\ref{subfig:sem1}). Now, consider the second position of $\pi$. Here, $\mathit{offer}$ does not hold. The set of traces with an observation equivalent prefix are all traces $\pi'$ such that $\Omega_a(\pi)[0,2] = \Omega_a(\pi')[0,2]$. For the unexplainable hiring system $\mathcal{N}$ we can now show that, no matter which pair of attributes $\alpha,\beta$ and corresponding counterfactual conditional $(\alpha \land \beta) \mcf_a \mathit{offer}$ we choose, there will always be an observation-equivalent trace such that the counterfactual does not hold (cf.~Subfigure~\ref{subfig:sem2} for the semantics of the counterfactual conditional). For example, assume we pick the pair $a_\mathit{sales}$ and $a_\mathit{f}$, i.e., attributes for Applicant that match the preference of Recruiter on the second position of trace $\pi$. The counterfactual conditional $(a_\mathit{sales} \land a_\mathit{f}) \mcf_a \mathit{offer}$ does in fact hold on $\pi$ at the second position, since there is the (unique) closest trace satisfying $a_\mathit{sales} \land a_\mathit{f}$ that also satisfies $\mathit{offer}$, namely:
\begin{align*}
\pi' = \{\}\{a_\mathit{sales}, a_\mathit{f},r_\mathit{sales},r_\mathit{f}, \mathit{offer}\} \{\}^\omega \enspace .
\end{align*}
However, there also exists an observation-equivalent trace such that the same counterfactual conditional does not hold. This is a trace where Recruiter picks a different preference at the second position. Since Recruiter's preference is unobservable by Applicant, this yields the following observation-equivalent trace
\begin{align*}
\pi'' = \{\}\{a_\mathit{it},a_\mathit{f},r_\mathit{accounting},r_\mathit{f}\} \{\}^\omega \enspace ,
\end{align*}
 that does not satisfy $(a_\mathit{sales} \land a_\mathit{f}) \mcf_a \mathit{offer}$, since the (unique) closest trace satisfying $a_\mathit{sales} \land a_\mathit{f}$ is:
 \begin{align*}
 \pi''' = \{\}\{a_\mathit{sales},a_\mathit{f},r_\mathit{accounting},r_\mathit{f}\} \{\}^\omega \enspace ,
 \end{align*}
 where $\mathit{offer}$ does not hold at the second position. The crux now is that in the unexplainable system we can find such an observation-equivalent trace for any counterfactual conditional in ICE's formula, since the preference of Recruiter is not observable by Applicant, and hence may be modified freely in observation-equivalent prefixes. In contrast, the same does not work in the explainable hiring system $\ekripke$ since Applicant can observe Recruiters preference and, hence, observation-equivalent traces are restricted to have the same preferences picked by Recruiter as in $\pi$. In particular, this means that $\pi''$ is not an observation-equivalent trace with respect to $\pi$ in the explainable system.
 
 \subsection{Flavors of Explainability}
 
 We have seen at the example of ICE how our logic uses the formalisms of counterfactual, epistemic and temporal logic to express a certain explainability requirement. Yet there are other conceivable notions of counterfactual explainability that can be specified in this logic. 
 
 \subsubsection{Weak Counterfactual Explainability} It may, for instance, not be necessary that Applicant knows the exact attributes which would have resulted in an offer, but instead only that there were some attributes that would have let to an offer. This is specified by the following formula:
\begin{align*}
 	\G \Big( \lnot \mathit{offer} \rightarrow  \K_a \left( \big(\bigvee_{\alpha,\beta \in \mathit{Att}(a)} (\alpha \land \beta) \big)\mcf_a \mathit{offer}\right)\Big) \enspace ,
\end{align*}
 which we term \emph{Weak Counterfactual Explainability} (WCE). Based on the semantics of the knowledge operator, it is easy to see that ICE is the strictly stronger requirement. This yields the following proposition.
 \medskip
 \begin{proposition}
 	ICE is strictly stronger than WCE, i.e., the models of ICE are a strict subset of WCE's models: $\mods(\mathit{ICE}) \subset \mods(\mathit{WCE})$.
 \end{proposition}
 
  \subsubsection{External Counterfactual Explainability} Both ICE and WCE require that Applicant is by themselves able to bring about the consequent of the counterfactual, and this is indeed the case in both the explainable and unexplainable hiring system presented in Section~\ref{subsec:hiring}. This can be used to formalize \emph{actionable} counterfactual explanations~\cite{PoyiadziSSBF20}, i.e., counterfactual explanations that range over only attributes under the control of the agent receiving the explanation. However, consider an alteration of the explainable hiring system where Applicant cannot obtain the qualifications for accounting, while this may still be Recruiter's preference. Hence, formally we modify the state space to obtain the modified Kripke structure $\kripke'$ follows:
  \begin{align*}
  	S' = \{(a_{\mathit{job}},a_{\mathit{gen}},r_{\mathit{job}},r_{\mathit{gen}}) \mid & \, r_{\mathit{job}} \in \{\mathit{accounting},\mathit{sales},\mathit{it}\} \land
  	 a_{\mathit{job}} \in \{\mathit{sales},\mathit{it}\}\\ &\land \; a_{\mathit{gen}},r_{\mathit{gen}} \in \{\mathit{m},\mathit{f}\} \} \cup \{s_0\} \enspace .
  \end{align*}
 The resulting hiring system $\mathcal{E}' = (\kripke',\Omega^\ekripke)$ does not satisfy ICE, as, e.g., none of the counterfactuals in the formula hold at the second position of $\pi''$ as defined in Section~\ref{subsec:example_traces}. This is because only Recruiter can induce the necessary change by changing their preference. Since Recruiter's preference is observable to Applicant, it may still be reasonable to include explanations that Applicant can deduce from these observations, but may be out of their control, i.e., \emph{external} explanations. This yields the following criterion which we term \emph{General Counterfactual Explainability} (GCE), which encompasses both internal and external explanations:
\begin{align*}
	\G \Big( \lnot \mathit{offer} \rightarrow  \big( \bigvee_{\alpha,\beta \in \mathit{Att}(a,r)} \K_a \left((\alpha \land \beta) \mcf_a \mathit{offer}\right)\big)\Big) \enspace ,
\end{align*}
where $\mathit{Att}(a,r) = \mathit{Att}^*(a)\cup\mathit{Att}^*(r)$.
Since the subformulas in the central disjunction of GCE subsume the ones present in ICE, it is again easy to deduce that the former is a strict relaxation of the latter. The strictness is witnessed by the modified hiring system $\mathcal{E}'$ discussed before.
\medskip
 \begin{proposition}
	ICE is strictly stronger than GCE.
\end{proposition}
\medskip

In this section, we have seen how our logic allows to formalize certain intricacies of different notions of explainability that pertain to questions such as: \emph{Does an agent know which exact actions explain some observed outcome? And are these actions solely under the control of the agent, or dependent on other agents, too?} The logic provides an ideal basis to formalize these intricacies and construct a taxonomy of explainability that discriminates between, e.g., weak and internal explainability. In the following sections we introduce more such notions of explainability based on our logic.

\section{A Tri-Modal Logic for Explainability}\label{sec:yltl}

 We now outline the formal semantics of the logic. We will use the shorthand \yltl to refer to the logic, which stands for \emph{whY Linear-time Temporal Logic}. The structure of the Y also represents that the logic is a fusion of three modal logics. First, we present the syntax and semantics of \yltl. Afterward, we will study the model-checking problem of \yltl. We then outline a decision procedure for finite-state model checking based on translating \yltl formulas into \foe.

\subsection{Syntax}

\yltl is an extension of \kltl with the original counterfactual operators~\cite{Lewis73} and counterfactuals for non-total similarity relations \cite{FinkbeinerS23}. This yields the following syntax for our logic \yltl:
\begin{align*}
	\varphi \Coloneqq~&p \mid \neg \varphi \mid \varphi \land \varphi \mid \LTLnext \varphi \mid \varphi \U \varphi \mid \K_a \, \varphi \mid \tag{\kltl}\\
	&\LTLnext^- \varphi \mid\LTLuntil^- \varphi \mid \tag{past-operators}\\
	&\varphi \cf_a \varphi \mid \varphi \ucf_a \varphi \tag{counterfactuals}
\end{align*}
where again $p \in \AP$ is an atomic proposition and $a \in \ags$ is an agent. \yltl inherits all of the derived operators of \kltl with past operators, as well as the counterfactual operators `Might' ($\varphi \mcf_a \psi \equiv \lnot (\varphi \cf_a \lnot \psi)$), a dual to `Would', and `Existential Might', a dual to `Universal Would' ($\varphi \emcf_a \psi \equiv \lnot (\varphi \ucf_a \lnot \psi)$).

We use Lewis' counterfactuals as predicates for causal reasoning because they are a common basis for a wide array of counterfactual causality definitions. While a more refined notion of causal predicates may be desirable, the literature is still divided on what refined notion generalizes to more than a few examples. Further, refined notions, such as actual causality, can often be encoded with counterfactuals~\cite{FinkbeinerS23}. We, therefore, hypothesize that an agent's desired explanation can always be expressed by Boolean combination of counterfactual dependencies with respect to the agent's similarity relation, which is covered by our logic. For instance, actual causality combines counterfactual reasoning with a minimality criterion~\cite{HalpernP05a}. We can require minimality of the counterfactual antecedent by enumerating all subformulas, i.e., for a specific antecedent $\alpha \land \beta$ we can extend the formula $(\alpha \land \beta) \mcf_a \mathit{offer}$ in ICE to:
$$ (\alpha \land \beta) \mcf_a \mathit{offer} \land \lnot (\alpha  \mcf_a \mathit{offer} ) \land \lnot (\beta  \mcf_a \mathit{offer} ) \enspace .$$

\subsection{Semantics}\label{subsec:semantics}

\yltl inherits the semantics of all shared operators from \kltl, such that we only need to define the semantics of the past-operators and counterfactuals. Since the semantics of counterfactuals rely on a similarity-based analysis, we need to extend the extended Kripke structures of \kltl further to accommodate for the agent's similarity relations. Hence, the semantics of a \yltl formula is defined with respect to an similarity-extended Kripke structure $\ekripke^+ = (\kripke,\Omega,\Sigma)$. Here, $\Sigma$ denotes the \emph{similarity map} $\Sigma : \ags \mapsto (\Pi \times \Pi \times \Pi \mapsto \mathcal{L}_{\text{\kltl}}(\AP \times \Pi))$ which provides a relational \kltl formula ranging over pairs of atomic propositions $\AP$ and trace variables $\Pi$.

\subsubsection{Similarity Map} The similarity map $\Sigma$ defines the similarity relations of the different agents, each with a (relational) \kltl formula. We first define the \emph{zipped trace} $z(\pi_1,\pi_2,\pi_3)$ of three traces $\pi_{1,2,3} \in (2^\AP)^\omega$ as follows for all $i \in \mathbb{N}$:
\begin{align*}
&z(\pi_1,\pi_2,\pi_3)[i] = \{(a,\pi_k) \in \AP \times \Pi \mid a \in \pi_k[i]\} \enspace .
\end{align*}
The zipped trace simply fuses the three traces together while enriching the atomic propositions with the information on which trace they originate from. This now allows us to evaluate the formula obtained from the similarity map on the zipped trace, as a way to characterize the underlying similarity relation. We denote the similarity relation of some agent $a \in \ags$ as $\Sigma_a$, and define it as:
$$ \Sigma_a = \{(\pi_1,\pi_2,\pi_3) \mid z(\pi_1,\pi_2,\pi_3) \models \Sigma(a)(\pi_1,\pi_2,\pi_3) \} \enspace .$$
Hence, three traces are related in the similarity relation of agent $a$ if and only if their zipped trace satisfies the formula specified by the similarity map for agent $a$. We require the similarity relation to satisfy some assumptions, which we specify for the two place relation $\Sigma_a^\pi = \{ (\pi_1,\pi_2) \mid (\pi,\pi_1,\pi_2) \in \Sigma_a  \}$ as in Lewis' original work~\cite{Lewis73}. Crucially, we allow $\Sigma_a^\pi$ to be non-total like in our recent extension~\cite{FinkbeinerS23}. With this, we ensure that subset-based similarity relations like the one described in Section~\ref{sec:motivation} and used, e.g., for actual causality~\cite{HalpernP05a}, can be handled by our logic. We require $\Sigma_a^\pi$ to be a preorder with $\pi$ as a minimum: $\forall \pi' : (\pi,\pi') \not\in \Sigma_a^\pi \rightarrow (\pi',\pi) \not\in \Sigma_a^\pi$, i.e., if a trace is not at least as far from $\pi$ as $\pi$ itself, it is not related to $\pi$ in $\Sigma_a^\pi$, and hence \emph{inaccessible}. We can also use the similarity relation to encode Lewis' notion of inaccessibility by simply not relating inaccessible traces, as we have relaxed it to a non-total relation which allows such non-ordered pairs. In the following section we will outline the consequences this relaxation has on the semantics of the counterfactual operators, and how to alleviate these with two additional operators with modified semantics.

\subsubsection{Counterfactuals} We can now proceed to specify the semantics of the counterfactual operators, for which we apply a similarity-based analysis~\cite{Lewis73}. Lewis defines counterfactuals as variably strict conditionals, which in multi-agent systems we interpret to mean that to hold on a specific trace, the consequent needs to hold in the closest accessible traces satisfying the antecedent. 
This now standard semantic treatment of counterfactuals in particular means that they cannot be expressed by a universal modality combined with a conditional, i.e., as Lewis argues, the semantics cannot be expressed with the usual universal modal operator. In our setting, this means that their semantics cannot be modeled with a knowledge operator and a conditional, i.e., $\K_a (\varphi_1 \rightarrow \varphi_2)$ is not equivalent to $\varphi_1 \cf \varphi_2$. Instead we define these counterfactuals in accordance with Lewis' original modal treatment. This results in the following semantics for a similarity-extended Kripke structure $\ekripke^+ = (\kripke,\Omega,\Sigma)$, an initial trace $\pi \in \Pi(\kripke)$, and a position $i$:
\begin{alignat*}{2}
	\ekripke^+,\pi,i \models \varphi_1\cf_a \varphi_2	~\text{ iff } &(1) \; \forall \pi' \in \Pi(\kripke): (\pi,\pi') \in \Sigma_a^\pi \rightarrow  \ekripke^+,\pi',i \nmodels \varphi_1 \; \lor\\ &(2) \; \exists \pi' \in \Pi(\kripke): (\pi,\pi') \in \Sigma_a^\pi \land \ekripke^+,\pi',i \models \varphi_1 \; \land\\ &\forall  \pi'' \in \Pi(\kripke): (\pi'',\pi') \in \Sigma_a^\pi \rightarrow \ekripke^+,\pi'',i \models (\varphi_1 \rightarrow \varphi_2) \enspace.
\end{alignat*}
Condition (1) represents a vacuity condition such that the `Would' counterfactual holds on a trace $\pi$ if there are no accessible traces where the antecedent $\varphi_1$ holds. It is easy to see how the quantification is restricted to traces that are related to $\pi$ in the similarity relation, i.e., traces that are accessible from $\pi$, through the implication after each universal quantifier and the conjunction after the existential quantifier. Condition (2), in principle, encodes the idea that on all closest counterfactual traces where $\varphi_1$ holds, $\varphi_2$ holds as well. 

\paragraph*{Infinite Chains of Closer Traces.} The complex nested quantification comes into play when there is no unique closest trace for some antecedent, as illustrated in the following example with an infinitely descending chain of progressively more similar traces. 
\medskip
\begin{example}\label{ex:limit}
	Consider the same similarity relation as used in the example from Section\label{ex:chains}~\ref{sec:motivation}, where $A = \mathit{Att}^+(a) \cup \mathit{Att}^+(r)$: 
    \begin{align*}
	\Sigma(a)(\pi,\pi_1,\pi_2) = \, &\LTLglobally \Big(\!\!\bigwedge_{p \in \A} \!\!\!\big((p,\pi) \not\leftrightarrow (p,\pi_1)\big) \rightarrow  \big((p,\pi) \not\leftrightarrow (p,\pi_2)\big)\Big) \, \land \\
	&\LTLglobally^- \Big(\!\!\bigwedge_{p \in \A} \!\!\!\big((p,\pi) \not\leftrightarrow (p,\pi_1)\big) \rightarrow  \big((p,\pi) \not\leftrightarrow (p,\pi_2)\big)\Big) \enspace ,
\end{align*} and the trace $\{p\}^\omega$. We are interested in the counterfactual $(\lnot \LTLglobally \LTLeventually p) \cf_a \top$, in a structure that contains all traces over the alphabet $\{p\}$. In this situation, we have an infinite chain of traces that satisfy $(\lnot \LTLglobally \LTLeventually p)$, i.e., $\{\}^\omega$, $\{p\}\{\}^\omega$, $\{p\}\{p\}\{\}^\omega$, etc. Hence, we cannot evaluate the consequent in a particular unique closest counterfactual trace, but instead need to make use of Lewis' elegant semantics for counterfactuals without the so-called limit assumption: We are looking for an accessible threshold trace $\pi'$, such that all at least as close traces $\pi''$ that satisfy the antecedent also satisfy the consequent.
\end{example} 
\medskip
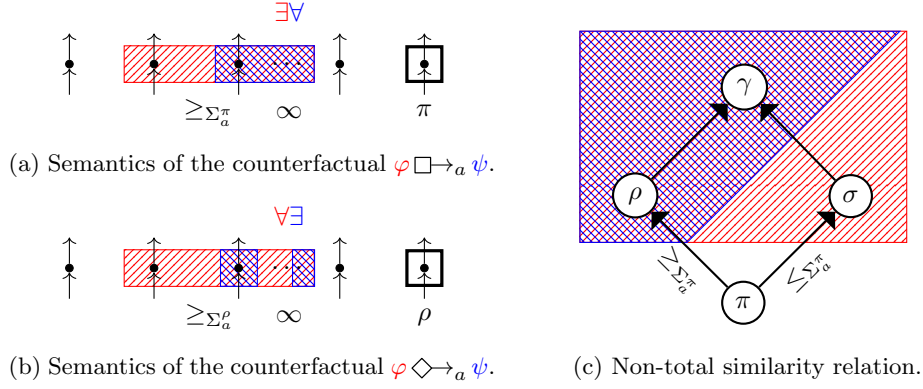
\begin{figure}[t]
    \centering
    \begin{subfigure}{.49\textwidth}
        \centering
    \begin{tikzpicture}[auto,
			node distance=0.35 and 1,every state/.style={minimum size=2pt,inner sep=1pt,fill},
            square/.style={regular polygon,minimum size=0.65cm,regular polygon sides=4}]

        \node[draw,rectangle,
                minimum width =2.5cm, 
                minimum height = 0.485cm,pattern=north east lines, pattern color=red,draw=red](s) at (-2.7,0) {};
                
        \node[draw,rectangle,
                minimum width =1.3cm, 
                minimum height = 0.485cm,pattern=north west lines, pattern color=blue,draw=blue](s) at (-2.1,0) {};

        \node[state](2){};
        \node[below = of 2](1){$\pi$};
        \node[above = of 2](3){};
        \path[->,draw] (1) edge (2)
        (2) edge (3);

        \node[state, left = of 2](2a){};
        \node[below = of 2a](1a){};
        \node[above = of 2a](3a){};
        \path[->,draw] (1a) edge (2a)
        (2a) edge (3a);
        
        \node[left = 0.25 of 2a](2b){\dots};
        
        \node[state, left = 0.25 of 2b](2c){};
        \node[below = of 2c](1c){};
        \node[above = of 2c](3c){};
        \path[->,draw] (1c) edge (2c)
        (2c) edge (3c);
        
        \node[state, left = of 2c](2d){};
        \node[below = of 2d](1d){};
        \node[above = of 2d](3d){};
        \path[->,draw] (1d) edge (2d)
        (2d) edge (3d);
        
        \node[state, left = of 2d](2e){};
        \node[below = of 2e](1e){};
        \node[above = of 2e](3e){};
        \path[->,draw] (1e) edge (2e)
        (2e) edge (3e);
    
        \node[draw,square,very thick](s) at (2) {};
        \node[below = 0.3 of 2b](x){$\infty$};
        \node[left = 0.3 of x]{$\geq_{\Sigma_a^\pi}$};
        \node[above = 0.3 of 2b]{\textcolor{red}{$\exists$}\textcolor{blue}{$\forall$}};
    \end{tikzpicture}
    \caption{Semantics of the counterfactual \textcolor{red}{$\varphi$}$\,\cf_a$\textcolor{blue}{$\,\psi$}.}\label{subfig:cf1a}
    
    \begin{tikzpicture}[auto,
			node distance=0.35 and 1,every state/.style={minimum size=2pt,inner sep=1pt,fill},
            square/.style={regular polygon,minimum size=0.65cm,regular polygon sides=4}]

        \node[draw,rectangle,
                minimum width =2.5cm, 
                minimum height = 0.485cm,pattern=north east lines, pattern color=red,draw=red](s) at (-2.7,0) {};
                
        \node[draw,rectangle,
                minimum width =0.485cm, 
                minimum height = 0.485cm,pattern=north west lines, pattern color=blue,draw=blue](s) at (-2.44,0) {};
                
        \node[draw,rectangle,
                minimum width =0.29cm, 
                minimum height = 0.485cm,pattern=north west lines, pattern color=blue,draw=blue](s) at (-1.59,0) {};
    
        \node[state](2){};
        \node[below = of 2](1){$\rho$};
        \node[above = of 2](3){};
        \path[->,draw] (1) edge (2)
        (2) edge (3);

        \node[state, left = of 2](2a){};
        \node[below = of 2a](1a){};
        \node[above = of 2a](3a){};
        \path[->,draw] (1a) edge (2a)
        (2a) edge (3a);
        
        \node[left = 0.25 of 2a](2b){\dots};
        
        \node[state, left = 0.25 of 2b](2c){};
        \node[below = of 2c](1c){};
        \node[above = of 2c](3c){};
        \path[->,draw] (1c) edge (2c)
        (2c) edge (3c);
        
        \node[state, left = of 2c](2d){};
        \node[below = of 2d](1d){};
        \node[above = of 2d](3d){};
        \path[->,draw] (1d) edge (2d)
        (2d) edge (3d);
        
        \node[state, left = of 2d](2e){};
        \node[below = of 2e](1e){};
        \node[above = of 2e](3e){};
        \path[->,draw] (1e) edge (2e)
        (2e) edge (3e);
    
        \node[draw,square,very thick](s) at (2) {};
        \node[below = 0.3 of 2b](x){$\infty$};
        \node[left = 0.3 of x]{$\geq_{\Sigma_a^\rho}$};
        \node[above = 0.3 of 2b]{\textcolor{red}{$\forall$}\textcolor{blue}{$\exists$}};
    \end{tikzpicture}
    \caption{Semantics of the counterfactual \textcolor{red}{$\varphi$}$\,\mcf_a$\textcolor{blue}{$\,\psi$}.}
    \label{subfig:cf1b}
    \end{subfigure}
    \begin{subfigure}{.49\textwidth}
        
    \centering
    \begin{tikzpicture}[auto,
			node distance=1 and 1,every state/.style={minimum size=2pt,inner sep=3pt},
            square/.style={regular polygon,minimum size=0.65cm,regular polygon sides=4}]
    
        \node[draw,rectangle,
                minimum width =4.3cm, 
                minimum height = 2.8cm,pattern=north east lines, pattern color=red,draw=red](s) at (0,2.2) {};
        \draw [draw=blue,pattern=north west lines,pattern color =blue]
       (-2.15,0.8) -- (-2.15,3.6) -- (2.06,3.6) -- (-0.74,0.8) -- cycle;
    \node[thick,state,fill=white](1){$\pi$};
    \node[thick,state,fill=white, above left = 1 and 1 of 1](2){$\rho$};
    \node[thick,state,fill=white, above right = 1 and 1 of 1](3){$\sigma$};
    \node[thick,state,fill=white, above right = 1 and 1 of 2](4){$\gamma$};
    \draw [thick,arrows = {-Stealth[inset=0pt, angle=90:7pt]}] (1) edge node[pos=0.1,rotate=-45] {$\geq_{\Sigma_a^\pi}$} (2);
    \draw [thick,arrows = {-Stealth[inset=0pt, angle=90:7pt]}] (1) edge node[swap,pos=0.1,rotate=45] {$\leq_{\Sigma_a^\pi}$} (3);
    \draw [thick,arrows = {-Stealth[inset=0pt, angle=90:7pt]}] (3) edge node[swap,pos=0.1,rotate=45] {} (4);
    \draw [thick,arrows = {-Stealth[inset=0pt, angle=90:7pt]}] (2) edge node[swap,pos=0.1,rotate=45] {} (4);
    \end{tikzpicture}
    \vspace{0.5em}
    \caption{Non-total similarity relation.}
    \label{subfig:cf2}
    \end{subfigure}
    \caption{Lewis' original semantics for the counterfactuals $\cf_a$ and $\mcf_a$ are illustrated in Subfigures~\ref{subfig:cf1a} and~\ref{subfig:cf1a}, respectively. Arrows and point depict traces that are ordered in ascending similarity to $\pi$ and $\rho$, respectively, according to the similarity relation $\geq_{\Sigma_a}$. Subfigure~\ref{subfig:cf2} highlights problems when evaluating the counterfactual \textcolor{red}{$\varphi$}$\,\cf_a$\textcolor{blue}{$\,\psi$} in a non-total similarity relation. Here, circles represent full traces such as $\phi$ or $\gamma$, while arrows indicate that two traces are ordered by the similarity relation $\leq_{\Sigma_a^\pi}$. In all subfigures, areas with diagonal lines (colored red) indicate that the covered traces satisfy \textcolor{red}{$\varphi$}, while crossed lines (colored red and blue) indicate that the traces satisfy \textcolor{blue}{$\psi$}.}
    \label{fig:cf}
\end{figure}

Figure~\ref{fig:cf} abstractly illustrates these semantics on infinite chains of closer traces. In Subfigure~\ref{subfig:cf1a}, we can see how the counterfactual $\varphi \cf_a \psi$ requires a continuous chain of traces satisfying $\psi$ as soon as we move up the similarity relation from traces that satisfy $\lnot \varphi$ to traces that satisfy $\varphi$, starting from the reference trace $\pi$. This is realized through the $\exists\forall$-quantifier alternation that requires a trace satisfying $\varphi$ such that all closer traces satisfying $\varphi$ also satisfy $\psi$. In contrast, the counterfactual $\varphi \mcf_a \psi$ requires for all traces satisfying $\varphi$ at least one closer trace satisfying $\psi$ -- even on infinite chains. This is realized through a $\forall\exists$-quantifier alternation and depicted in Subfigure~\ref{subfig:cf1b}.

\paragraph*{Non-Total Similarity Relations.} Unlike in Lewis' original account, we allow a similarity relation $\Sigma_a^\pi$ of some agent $a$ to be non-total. As a consequence, Lewis' original semantics yield some rather unintuitive inferences~\cite{FinkbeinerS23}, which we illustrate in the following example.
\medskip
\begin{example}\label{ex:chains}
	Consider again the similarity relation as defined in Section~\ref{sec:motivation} and recalled in the previous Example\label{ex:limit} and the trace $\pi = \{\}^\omega$, with the counterfactual $(p \lor q) \cf_a p$, in a structure that contains all traces over the alphabet $\{p,q\}$. We depict $\pi$ and three other traces with their comparative similarity in Subfigure~\ref{subfig:cf2}. The counterfactual is satisfied by $\pi$, as we have the closest counterfactual trace $\rho = \{p\}\{\}^\omega$ as a witness for the existential quantifier in Condition~2 of the semantics of $\cf_a$. However, this does not match the intended semantics of the `Would' counterfactual. The counterfactual is supposed to express that the consequent $p$ holds on all closest counterfactual traces. However, there is the closest counterfactual trace $\sigma = \{q\}\{\}^\omega$ that does not satisfy $p$.
\end{example}
\medskip
The problem with Lewis' original semantics in non-total similarity relations is that the existential quantifier in Condition 2 implicitly also automatically quantifies existentially over the unrelated chains of at least as similar traces in the similarity relation. In Example~\ref{ex:chains}, these (in this case finite, but in general possibly infinite) chains are, on the one chain, the trace changing $p$ and, on the other chain, the trace changing $q$. These traces with single changes are incomparable with each other regarding their similarity to the reference trace $\{\}^\omega$, while the trace that changes both $p$ and $q$ is comparable to both of these traces that change only single atomic propositions (it is, of course, less similar to $\{\}^\omega$ than both of the traces). However, since the trace with the single changes already satisfy the antecedent of the counterfactual, they have to be considered as possible threshold traces for Lewis' criterion. However, the implicit existential quantification allows the semantics to ignore whole chains (in this case the chain with $\{q\}\{\}^\omega$), which then do not need a threshold trace satisfying Lewis' criterion. 

In earlier work~\cite{FinkbeinerS23}, we proposed an alternative counterfactual operator which we include in \yltl. The operator is called `Universal Would' counterfactual, because the semantics are based on universal quantification over the chains of the similarity relation as follows, again for a similarity-extended Kripke structure $\ekripke^+ = (\kripke,\Omega,\Sigma)$, an initial trace $\pi \in \Pi(\kripke)$, and a position $i$:
\begin{alignat*}{2}
	\ekripke^+,\pi,i \models \varphi_1\ucf_a \varphi_2	~\text{ iff } \, &(1) \; \forall \pi' \in \Pi(\kripke): (\pi,\pi') \in \Sigma_a^\pi \land \ekripke^+,\pi',i \models \varphi_1 \, \rightarrow\\ &(2) \; \exists \pi'' \in \Pi(\kripke):(\pi'',\pi') \in \Sigma_a^\pi \land\ekripke^+,\pi'',i \models \varphi_1 \, \land\\ &\forall  \pi''' \in \Pi(\kripke):(\pi''',\pi'') \in \Sigma_a^\pi \rightarrow \ekripke^+,\pi''',i \models (\varphi_1 \rightarrow \varphi_2) \enspace .
\end{alignat*}
Intuitively, this operator lifts Lewis' semantics of the `Would' operator to non-total similarity relations by applying it to every chain in the relation. This is achieved by prepending the non-vacuous condition of Lewis' definition for $\cf$, i.e., Condition~2, with a universal quantification that effectively quantifies over chains of traces (Condition~1). The semantics of Lewis' vacuity condition is then also directly captured by the initial universal quantification and the implication, such that we do not need the same disjunction as in Lewis' definition for $\cf$. For every chain with at least one counterfactual world not satisfying $\varphi_1$, the same requirement as posed by Lewis' original `Would' counterfactual has to hold: the threshold trace is $\pi''$ bound by the existential quantifier, and this threshold trace has to be found on the same chain that the universally quantified $\pi'$ is on. Consequently, there has to be a threshold trace on every chain containing a trace that satisfies $\varphi_1$ that is accessible from $\pi$. Local vacuity is still allowed, i.e., a whole chain without a single trace satisfying $\varphi_1$ does not need a closest trace satisfying $\varphi_2$. 
\medskip
\begin{example}
    To see how this semantics fixes the problem raised in Example~\ref{ex:chains}, consider again the trace $\{\}^\omega$, with the counterfactual $(p \lor q) \cf_a p$, in a structure that contains all traces over the alphabet $\{p,q\}$ and under the same similarity relation used in Section~\ref{sec:motivation} (with this new alphabet). The problem is that Lewis' existential quantifier allowed us to choose between the traces $\{p\}\{\}^\omega$ and $\{q\}\{\}^\omega$ as a witnessing counterfactual world. However, if we use the stronger `Universal Would' operator $(p \lor q) \ucf_a p$, the universal quantifier requires us to find a witnessing counterfactual operator on every chain. Since there is no at least as close trace $\pi''$ that satisfies $p$ for the trace $\{q\}\{\}^\omega$, we have that $(p \lor q) \ucf_a p$ is not satisfied on the trace $\{\}^\omega$ in this scenario. This is as desired, because $p$ does not hold on all of the closest traces satisfying $p \lor q$.
\end{example}

\subsubsection{Agent-Specific Similarity} A key feature of our counterfactuals is that their semantics are defined with respect to a \emph{specific} agent's similarity relation. This represents that agents may have different internal models about the causal workings of the system. For instance, in the hiring system described in Section~\ref{sec:motivation} it may be sensible that Applicant does not consider counterfactual scenarios where their gender attribute is different from the actual trace. An explanation based on such counterfactuals would not be actionable~\cite{PoyiadziSSBF20}, and may hence be undesired in many cases. Actionable explanations range only over actions and attributes that are fully under control of the agent receiving the explanations, which clearly is not the case for the gender attribute. Whether an antecedent is actionable or not highly depends on the scenario and the agent at hand, which motivates our flexible formalism of agent-specific similarity relations. With \yltl, such requirements can then be encoded by making, e.g., the traces with a modified atomic proposition $a_{gen}$ inaccessible from the reference trace in the similarity map $\Sigma'$ for agent $a$ as follows:
\begin{align*}
	\Sigma'(a)(\pi,\pi_1,\pi_2) = \Sigma(a)(\pi,\pi_1,\pi_2) &\land \LTLglobally^- (a_{gen},\pi) \leftrightarrow (a_{gen},\pi_1) \land (a_{gen},\pi) \leftrightarrow (a_{gen},\pi_2) \\
	&\land \, \LTLglobally (a_{gen},\pi) \leftrightarrow (a_{gen},\pi_1) \land (a_{gen},\pi) \leftrightarrow (a_{gen},\pi_2) \enspace .
\end{align*}
Here, we use the previous similarity relation $\Sigma(a)(\pi,\pi_1,\pi_2)$ defined in Section~\ref{subsec:hiring}. Recall that the idea of that relation was that the changes between the actual trace $\pi$ and the closer trace $\pi_1$ are a subset of the changes between the actual trace $\pi$ and the farther trace $\pi_2$. $\Sigma'(a)(\pi,\pi_1,\pi_2)$ now requires $a_{gen}$ to be the same on all three traces, which means traces of the system that change $a_{gen}$ are not related, hence inaccessible. This makes explainability specifications such as ICE harder to satisfy. 

Yet, other agents such as Recruiter may still consider counterfactual traces where the attribute $a_{gen}$ is modified, i.e., their similarity relation is $ \Sigma'(r) =  \Sigma(a)$. As a result we have that the explainable system $\mathcal{E} = (\kripke,\Omega^\ekripke,\Sigma')$ where all atomic propositions are observable by both agents does not satisfy ICE from the point of view of Applicant, but does satisfy ECE, i.e., \emph{External Counterfactual Explainability}, from the point of view of Recruiter, where ECE is formalized as follows:
\begin{align*}
	\G \Big( \lnot \mathit{offer} \rightarrow  \big( \bigvee_{\alpha,\beta \in \mathit{Att}(a)} \K_r \left((\alpha \land \beta) \mcf_r \mathit{offer}\right)\big)\Big) \enspace .
\end{align*}
Note that since both agents can observe the same atomic propositions, and hence $\K_a$ and $\K_r$ are in principle interchangeable, this difference is completely due to the fact that there are some $\alpha,\beta \in \mathit{Att}(a)$ for every position $i$ such that the counterfactual conditional $(\alpha \land \beta) \mcf_a \mathit{offer}$ holds, while this is not the case for the counterfactual $(\alpha \land \beta) \mcf_r \mathit{offer}$ that refers to the similarity relation of Recruiter.

\subsection{Model Checking}

In this section, we develop an approach to automatically verify whether a given system satisfies a \yltl specification. Our results apply to systems defined by finite similarity-extended Kripke structures. Under this assumption, we can then show the decidability of the \yltl model-checking problem by reducing it to model checking of an equivalent formula in Extended Monadic First-Order Logic (\foe). This is a decidable problem~\cite{FinkbeinerZ17,CoenenFHH19}, and we now outline this logic as a preliminary.

\paragraph{Extended Monadic First-Order Logic}
\foe is the monadic first-order logic of order (\fo) extended with the equal-level predicate~$E$~\cite{FinkbeinerZ17} for expressing hyperproperties~\cite{ClarksonS10}, i.e., properties that relate multiple executions of a system to one another. For a predefined set $V$ of first-order variables, the syntax of \foe is defined by the following grammar:
\begin{align*}
	\varphi &\Coloneqq \psi \mid \neg \varphi \mid \varphi \lor \varphi \mid \exists x .\ \varphi\\
	\psi &\Coloneqq P_p(x) \mid x < y \mid x = y \mid E(x,y) \enspace , 
\end{align*}
where $p \in \AP$ is an atomic proposition and $x,y \in V$ are first-order variables. An \foe formula is \emph{closed} when all variables are bound by a quantifier. \foe formulas are interpreted over a set of traces~$\Pi$. The first-order variables range over the domain $\Pi \times \mathbb{N}$. The order~$<$ is now only interpreted over variables referring to the same trace: $< \Coloneqq \{((\pi,n_1),(\pi,n_2)) \in (\Pi \times \mathbb{N})^2 \mid n_1 < n_2\}$. The equal-level predicate holds if two variables refer to the same position in (possibly) two different traces: $E \Coloneqq \{((\pi_1,n),(\pi_2,n)) \in (\Pi \times \mathbb{N})^2\}$. The predicate $P_p$ encodes the truth-value of atomic propositions: $P_p \Coloneqq \{(\pi,n) \mid p \in \pi[n]) \}$. We say that a closed \foe formula $\varphi$ is satisfied by an extended Kripke structure $\ekripke$, denoted by $\ekripke \models \varphi$, iff $\varphi$ interpreted over $\Pi(\kripke)$ is true.

We now outline our result on \yltl model checking. This utilizes a translation function $\tofo$ described in the proof of the following Lemma. The translation mirrors the idea used to translate \ltl into first-order logic defined in Kamp's seminal theorem~\cite{Kamp68}, which we extend for the knowledge operator and the counterfactuals, and for this we use the equal-level predicate provided by \foe.
\medskip
\begin{lemma}\label{lem:yltltofoe}
	For every \yltl formula $\varphi$ there exists a formula in \foe $\varphi'$ that characterizes the same set of models, i.e., such that $\mods(\varphi) = \mods(\varphi')$.
\end{lemma}

\begin{proof}
	The proof relies on a linear translation $\tofo$ from \yltl to \foe. We will use syntactic sugar for successors and minimal positions: $\suc(x,y) \Coloneqq x < y \land \lnot \exists z . \, x < z < y$ and $\mini(x) \Coloneqq \lnot \exists y . \, \suc(y,x)$. 
	In the end, the \foe formula proving our claim is obtained from $\varphi$ by 
	\begin{align}
 \tofo(\varphi) \Coloneqq \forall x_0 . \, \mathit{min}(x_0) \rightarrow \tofo(\varphi,x_0) \enspace . \label{eq:proof}
 \end{align}
The \foe formula $\tofo(\varphi,x_0)$ is constructed inductively based on the current sub-formula of the \yltl formula $\varphi$ (ranging over a set $\AP$) and the current time-point of interest encoded in the second argument, which is initially $x_0$ but may change through trace quantification from, e.g., epistemic operators. Recall that the first-order variables of  \foe are in fact tuples $(\pi,n) \in \Pi \times \mathbb{N}$ of a trace variable and a position, let us denote for some tuple $x_t = (\pi,n)$: $x_t|_1 = \pi$ and $x_t|_2 = n$ for projecting to the components of the tuple. Note that $\varphi$ technically includes atomic propositions from the set $\AP \cup (\AP \times \Pi)$ since we also need to translate the \kltl formulas obtained from the similarity map, which range over tuples of atomic propositions and trace variables. These tuples from $\AP \times \Pi$ are unrelated to the first-order variables $(\pi,n) \in \Pi \times \mathbb{N}$ and become relevant only when translating counterfactual operators. We start with the simpler cases, for which the construction of the \foe formula is as follows.
	\begin{alignat*}{2}
		&\tofo((p,\pi),x_t)	&	~=~ &P_p((\pi,x_t|_2))\\[0.5ex]
		&\tofo(p,x_t)	&	~=~ &P_p(x_t)\\[0.5ex]
		&\tofo(\lnot \varphi,x_t)	&	~=~  &\lnot \tofo(\varphi,x_t)\\[0.5ex]
		&\tofo(\varphi_1 \lor \varphi_2, x_t)	&	=~  &\tofo(\varphi_1,x_t) \lor \tofo(\varphi_2,x_t)\\[0.5ex]
		&\tofo(\LTLnext \varphi,x_t)	&	=~  &\exists x_t^+ . \, \suc(x_t,x_t^+) \land \tofo(\varphi,x_t^+)\\[0.5ex]
		&\tofo(\LTLnext^- \varphi,x_t)	& =~  &\exists x_t^- . \, \suc(x_t^-,x_t) \land \tofo(\varphi,x_t^-)\\[0.5ex]
		&\tofo(\varphi_1 \LTLuntil \varphi_2,x_t)	&	=~  &\exists x_2 \geq x_t. \, \tofo(\varphi_2,x_2) \, \land (\forall x_1. \, x_t \leq x_1 < x_2 \rightarrow \tofo(\varphi_1,x_1))\\[0.5ex]
		&\tofo(\varphi_1 \LTLuntil^- \! \varphi_2,x_t)	&	=~  &\exists x_2^- \leq x_t. \, \tofo(\varphi_2,x_2^-) \, \land (\forall x_1^-. \, x_t \geq x_1^- > x_2 \rightarrow \tofo(\varphi_1,x_1^-))\\[0.5ex]
		&\tofo(\K_a \, \varphi,x_t)	& =~  &\forall x_e. \, E(x_e,x_t) \land (\forall x_e^-,x_t^- . \, x_e^- \leq x_e \land \, x_t^- \leq x_t \land E(x_e^-,x_t^-) \rightarrow\\
		& & & \!\!\!\!\!\!\!\!\bigwedge_{p \in \Omega(a)} P_p(x_e^-) \leftrightarrow P_p(x_t^-))\rightarrow \tofo(\varphi,x_e)
	\end{alignat*}
The most involved formulas are obtained from translating the counterfactual operators. Note that the expression $\Sigma(a)(x_t|_1,x_t|_1,x_e|_1),x_t)$ that appears throughout the formulas simply denotes the \kltl formula characterizing the similarity relation of agent $a$, where the parameters are in this case instantiated by the trace variables of $x_t$ (twice) and of $x_e$. This double instantiation results from encoding accessibility via the similarity relation. The translation for the counterfactuals proceeds as follows.
		\begin{alignat*}{2}
		\tofo(\varphi_1 \cf_a \varphi_2,x_t)	~=~ & & &(\forall x_e . \, E(x_e,x_t) \land \tofo(\Sigma(a)(x_t|_1,x_t|_1,x_e|_1),x_t) \rightarrow \lnot \tofo(\varphi_1,x_e))\\
		& & &  \lor \exists x_e. E(x_e,x_t) \land \tofo(\Sigma(a)(x_t|_1,x_t|_1,x_e|_1),x_t)\ \land \tofo(\varphi_1,x_e)\\
		& & & \land \forall x_c . \tofo(\Sigma(a)(x_t|_1,x_c|_1,x_e|_1),x_t) \rightarrow (\tofo(\varphi_1,x_c) \rightarrow \tofo(\varphi_2,x_c))\\[1ex]
		\tofo( \varphi_1 \ucf_a \varphi_2,x_t)	~=~ & & & (\forall x_a . \, E(x_a,x_t) \land \tofo(\Sigma(a)(x_t|_1,x_t|_1,x_a|_1),x_t)\land \tofo(\varphi_1,x_a)\\
		& & &  \rightarrow \exists x_e. E(x_e,x_a) \land \tofo(\Sigma(a)(x_t|_1,x_e|_1,x_a|_1),x_t)  \land \tofo(\varphi_1,x_e)\\
		& & & \land \forall x_c . \tofo(\Sigma(a)(x_t|_1,x_c|_1,x_e|_1),x_t) \rightarrow (\tofo(\varphi_1,x_c) \rightarrow \tofo(\varphi_2,x_c))
	\end{alignat*}
Note that this function is well-defined only because we do not allow formulas from the similarity map to themselves include counterfactuals. Otherwise, the translation function could include a circular dependency where translating a counterfactual operator requires translating a similarity relation which in turn again requires translating a counterfactual and so on.

In the end, the equivalence between $\tofo(\varphi)$ and $\varphi$ (cf.~Equation~\ref{eq:proof}) can be shown by structural induction over $\varphi$.
\end{proof}

While Lemma~\ref{lem:yltltofoe} alone is just a statement about comparative expressiveness, it also indirectly provides us with an algorithm that given a system as a finite Kripke structure and a specification as a \yltl formula automatically verifies whether the systems satisfies the formula. This is quite remarkable as previous results for model checking logics that combine temporal operators and counterfactuals only considered counterfactuals as top-level operators~\cite{FinkbeinerS23}. Compared to this work, we lose the ability to express $\omega$-regular temporal properties, but gain the ability to nest counterfactuals and temporal operators, and additionally include knowledge operators. While nesting counterfactuals is mostly of theoretical interest, the additional knolwedge operators are crucial for expressing explainability. To the best of our knowledge, we present the first algorithm for model checking a logic that combines temporal operators and counterfactuals arbitrarily.
\medskip
\begin{theorem}[\yltl Model Checking]\label{thm:mc}
	There is an algorithm that, given a finite extended Kripke structure $\ekripke$ and a \yltl formula $\varphi$, checks whether $\ekripke \models \varphi$.
\end{theorem}

\begin{proof}
	In Lemma~\ref{lem:yltltofoe} we have shown that we can construct an equivalent \foe-formula $\varphi'$ for the \yltl formula $\varphi$. Since \foe is strictly less expressive than HyperQPTL (LTL with trace and propositional quantifiers)~\cite{CoenenFHH19}, and there is a model-checking algorithm for HyperQPTL~\cite{Rabe16}, the claim follows immediately.
\end{proof}

The existence of a model-checking algorithm is what makes our logic useful in practice: Not only is it possible to express several notions of explainability, it is also possible to automatically verify them. While an exact complexity analysis of model checking \yltl is out of scope of this paper, it should be noted that the complexity of the current encoding is non-elementary, with the tower of exponents scaling with the number of nested counterfactuals and knowledge operators. However, this is not worse than the complexity of model checking just knowledge and temporal operators~\cite{BozzelliMM24}. Moreover, in practice, we are mostly concerned with formulas that have only a few nested operators, as is the case for all of the explainability requirements formalized in this work.

\subsection{Side Result on the Expressiveness of KLTL}

Besides providing an algorithm for model checking \yltl formulas, Lemma~\ref{lem:yltltofoe} also includes a translation from \kltl to \foe that was loosely described earlier by Hofmann~\cite{Hofmann22}. This translation shows that \foe subsumes \kltl. We now outline how we can combine this with a result from Bozzelli et al.~\cite{BozzelliMP15} to show that \kltl is strictly less expressive than \foe. For completeness, we also recall some results regarding the comparative expressiveness of \kltl and HyperLTL (\ltl with quantification over traces), as well as HyperQPTL (HyperLTL with propositional quantifiers). As a reference, the syntax of HyperQPTL is build according to the following grammar:
$$
	\psi \Coloneqq \exists \pi \ldot \psi \mid \forall \pi \ldot \psi \mid \exists p \ldot \psi \mid \forall p \ldot \psi \mid \varphi' \enspace ,
$$
where $p$ is a fresh atomic proposition and $\pi$ is a trace variable. $\varphi'$ is an \ltl formula, i.e., build according to the grammar of \kltl (Section~\ref{subsec:logics}) without the knowledge operator and past-time operators. The syntax of HyperLTL can be obtained by removing $\forall p \ldot \psi $ and $\exists p \ldot \psi$ from the above grammar of HyperQPTL.

Previously, Bozzelli et al.~\cite{BozzelliMP15} have shown that \kltl's expressiveness is incomparable to the expressiveness of HyperLTL. Rabe~\cite{Rabe16} showed that \kltl can be encoded in HyperQPTL and Hofmann described that this encoding can be adapted for \foe, which we have confirmed in the proof of Lemma~\ref{lem:yltltofoe}. Combined, these results mean that \kltl lies strictly further down in the hierarchy of hyperlogics.
\medskip
\begin{theorem}\label{thm:exp}
	\!\foe is strictly more expressive than \kltl.
\end{theorem}

\begin{proof}
	Lemma~\ref{lem:yltltofoe} shows that \foe is at least as expressive as \yltl, and since \yltl subsumes \kltl trivially, it follows that \foe is at least as expressive as \kltl. It therefore only remains to show strictness. Strictness follows from previous results: (1) the proof that \kltl does not subsume HyperLTL presented by Bozzelli et al.~\cite{BozzelliMP15} and (2) by the subsumption of HyperLTL through \foe shown by Coenen et al.~\cite{CoenenFHH19}, as follows:
	(1) Bozzelli et al.\ provide the following HyperLTL formula that cannot be expressed in \kltl:
	$$\varphi_H = \exists \pi \ldot \exists \pi' \ldot p_\pi \LTLuntil \big((p_\pi \land \lnot p_{\pi'}) \land \LTLnext \LTLglobally (p_\pi \leftrightarrow p_{\pi'})\big) \enspace.$$
	The intuition behind their proof is that \kltl cannot compare two different traces at an unbounded number of positions. We refer to the full version~\cite{BozzelliMP14} of Bozzelli et al.'s paper for the detailed proof.
	With (2), we know that there exists an \foe formula $\varphi_{\mathit{fo}}$ that is equivalent to $\varphi_H$. $\varphi_{\mathit{fo}}$ is then not expressible in \kltl, which proves the claimed strictness of the inclusion.
\end{proof}

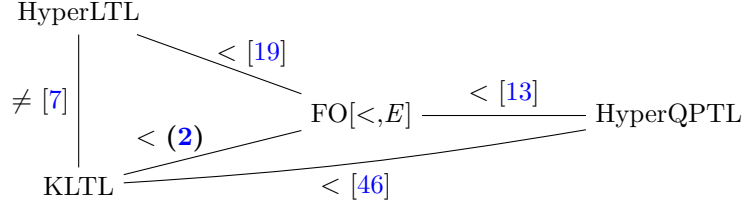
\begin{figure}
\centering
	\begin{tikzpicture}[node distance=6.5em,draw]
		\node[] (hyperltl) {HyperLTL};
		\node[below of=hyperltl] (kltl) {\kltl};
		\node[below right of=hyperltl,xshift = 6em,yshift = 0.75em] (foe) {\foe};
		\node[right of=foe,xshift = 5em] (hyperqptl) {HyperQPTL};
		\draw (hyperltl)  -- (kltl) node[midway,left]{$\neq$~\cite{BozzelliMP15}};
		\draw (hyperltl)  -- (foe) node[pos=0.7,above]{$<$~\cite{FinkbeinerZ17}};
		\draw (foe)  -- (hyperqptl) node[midway,above]{$<$~\cite{CoenenFHH19}};
		\draw (kltl)  edge[bend right=3] node[midway,below]{$<$~\cite{Rabe16}} (hyperqptl) ;
		\draw (kltl)  -- (foe) node[pos=0.3,above,xshift=-0.25em]{$<$~\textbf{(\ref{thm:exp})}};
	\end{tikzpicture}
 
	\caption{\kltl's exact place in the hierarchy of hyperlogics. The result of Theorem~\ref{thm:exp} is highlighted in bold.}
\end{figure}
\section{Related Work}

There is a long line of works on combining modal logics~\cite{sep-logic-combining}. In this section we focus only on works related to combinations of counterfactual, epistemic and temporal operators, which have been combined in pairs for a variety of applications. A connection between knowledge and counterfactual dependencies in the situation calculus has been drawn by Khan and Lespérance~\cite{KhanL21}. This has been extended to define \emph{explanations} for agent behavior~\cite{DBLP:conf/ecai/KhanR23}, in particular accounting for theory-of-mind reasoning. Contrary to these works, we focus on \emph{explainability} as a system property and provide an approach for verification, but we also appeal to theory-of-mind reasoning with our agent-specific similarity relations, which allow to model the internal mental states of the agents. Knowledge and causality have been combined to reason about deceptive AI~\cite{Sakama21}. Counterfactuals and the knowledge modality have also been combined to express hypothetical knowledge~\cite{Halpern99} and rationality~\cite{Sandu21,Stalnaker06} in game theory.  Liu and Lorini~\cite{LiuL23} study modal logics for defining individual explanations for classifiers, and Aguilera-Ventura et al.~\cite{Aguilera-Ventura23} have recently studied grounding similarity relations for counterfactuals.

Besides counterfactual epistemic logics, our work also builds on a long line of research into logics that reason about knowledge and time, which originated in the analysis of distributed protocols~\cite{LadnerR86,FaginHMV1995} and have been applied to a variety of applications such as information-flow security~\cite{MeydenS04,BalliuDG11,HalpernO08}, as well as knowledge-based programs~\cite{MeydenV98}. Counterfactual and temporal reasoning has been combined to reason about temporal aspects of causality~\cite{CoenenFFHMS22,FinkbeinerFMS24,CarelliFS25,ZiemekPFJB22}.

Our work studies the epistemics of explainability and abstracts away from questions such as how to visualize explanations, and what explanations are relevant for a human user in a given context. There is a variety of works that study these orthogonal questions~\cite{KohlBLOSB19,HorakCMHFMDFD22,BrandaoMMLC22,Miller19,LangerOSHKSSB21,SchlickerLOBKW21}. Moreover, there are several works on generating explanations for more complex system architectures~\cite{AudemardKM20,DarwicheJ22}.

\section{Conclusion \& Outlook}
We have studied a logic that combines the long-studied modal operators of counterfactual, epistemic and temporal logics for specification and verification of explainability requirements. We have demonstrated how the logic can be used to define the first formal taxonomy of counterfactual explainability that encompasses the notions of internal, external, general, and weak explainability. We believe this aspect of our study can be spun much further by introducing additional features to the logic, for instance minimality constraints on counterfactual antecedents~\cite{FinkbeinerS23}, or by considering combinations of counterfactual and probabilistic reasoning~\cite{ZiemekPFJB22} as explanatory properties. As another aspect, we have proven that the \yltl model-checking problem is decidable for finite-state multi-agent systems. We plan on building on this result by developing practical model-checking tools for explainability requirements. On the theoretical side, we have made first steps toward analyzing the expressivity of the combined logic in relation to other hyperlogics. These are also, to the best of our knowledge, the first results on model checking and expressivity of counterfactual operators when combined arbitrarily with temporal operators. Building on these results, we have recently proposed an approach for analyzing explainability and privacy tradeoffs in multi-agent systems, which uses a second-order version of \yltl to enable quantification over arbitrary counterfactual antecedents~\cite{FinkbeinerFS25}.

\bibliography{sn-bibliography}

\begin{thebibliography}{54}
\providecommand{\natexlab}[1]{#1}
\providecommand{\url}[1]{{#1}}
\providecommand{\urlprefix}{URL }
\providecommand{\doi}[1]{\url{https://doi.org/#1}}
\providecommand{\eprint}[2][]{\url{#2}}
 \bibcommenthead

\bibitem[{Aguilera{-}Ventura et~al(2023)Aguilera{-}Ventura, Herzig, Liu, and Lorini}]{Aguilera-Ventura23}
Aguilera{-}Ventura C, Herzig A, Liu X, et~al (2023) Counterfactual reasoning via grounded distance. In: Marquis P, Son TC, Kern{-}Isberner G (eds) Proceedings of the 20th International Conference on Principles of Knowledge Representation and Reasoning, {KR} 2023, Rhodes, Greece, September 2-8, 2023, pp 2--11, \doi{10.24963/KR.2023/1}, \urlprefix\url{https://doi.org/10.24963/kr.2023/1}

\bibitem[{Almagor and Lahijanian(2020)}]{AlmagorL20}
Almagor S, Lahijanian M (2020) Explainable multi agent path finding. In: Seghrouchni AEF, Sukthankar G, An B, et~al (eds) Proceedings of the 19th International Conference on Autonomous Agents and Multiagent Systems, {AAMAS} '20, Auckland, New Zealand, May 9-13, 2020. International Foundation for Autonomous Agents and Multiagent Systems, pp 34--42, \doi{10.5555/3398761.3398771}, \urlprefix\url{https://dl.acm.org/doi/10.5555/3398761.3398771}

\bibitem[{Audemard et~al(2020)Audemard, Koriche, and Marquis}]{AudemardKM20}
Audemard G, Koriche F, Marquis P (2020) On tractable {XAI} queries based on compiled representations. In: Calvanese D, Erdem E, Thielscher M (eds) Proceedings of the 17th International Conference on Principles of Knowledge Representation and Reasoning, {KR} 2020, Rhodes, Greece, September 12-18, 2020, pp 838--849, \doi{10.24963/kr.2020/86}, \urlprefix\url{https://doi.org/10.24963/kr.2020/86}

\bibitem[{Balliu et~al(2011)Balliu, Dam, and Guernic}]{BalliuDG11}
Balliu M, Dam M, Guernic GL (2011) Epistemic temporal logic for information flow security. In: Askarov A, Guttman JD (eds) Proceedings of the 2011 Workshop on Programming Languages and Analysis for Security, {PLAS} 2011, San Jose, CA, USA, 5 June, 2011. {ACM}, p~6, \doi{10.1145/2166956.2166962}, \urlprefix\url{https://doi.org/10.1145/2166956.2166962}

\bibitem[{Beckers(2022)}]{Beckers22}
Beckers S (2022) Causal explanations and {XAI}. In: Sch{\"{o}}lkopf B, Uhler C, Zhang K (eds) 1st Conference on Causal Learning and Reasoning, CLeaR 2022, Sequoia Conference Center, Eureka, CA, USA, 11-13 April, 2022, Proceedings of Machine Learning Research, vol 177. {PMLR}, pp 90--109, \urlprefix\url{https://proceedings.mlr.press/v177/beckers22a.html}

\bibitem[{Bozzelli et~al(2014)Bozzelli, Maubert, and Pinchinat}]{BozzelliMP14}
Bozzelli L, Maubert B, Pinchinat S (2014) Unifying hyper and epistemic temporal logic. CoRR abs/1409.2711. \urlprefix\url{http://arxiv.org/abs/1409.2711}, {\href{https://arxiv.org/abs/1409.2711}{{1409.2711}}}

\bibitem[{Bozzelli et~al(2015)Bozzelli, Maubert, and Pinchinat}]{BozzelliMP15}
Bozzelli L, Maubert B, Pinchinat S (2015) Unifying hyper and epistemic temporal logics. In: Pitts AM (ed) Foundations of Software Science and Computation Structures - 18th International Conference, FoSSaCS 2015, Held as Part of the European Joint Conferences on Theory and Practice of Software, {ETAPS} 2015, London, UK, April 11-18, 2015. Proceedings, Lecture Notes in Computer Science, vol 9034. Springer, pp 167--182, \doi{10.1007/978-3-662-46678-0\_11}, \urlprefix\url{https://doi.org/10.1007/978-3-662-46678-0\_11}

\bibitem[{Bozzelli et~al(2024)Bozzelli, Maubert, and Murano}]{BozzelliMM24}
Bozzelli L, Maubert B, Murano A (2024) On the complexity of model checking knowledge and time. {ACM} Trans Comput Log 25(1):8:1--8:42. \doi{10.1145/3637212}, \urlprefix\url{https://doi.org/10.1145/3637212}

\bibitem[{Brandao et~al(2022)Brandao, Mansouri, Mohammed, Luff, and Coles}]{BrandaoMMLC22}
Brandao M, Mansouri M, Mohammed A, et~al (2022) Explainability in multi-agent path/motion planning: User-study-driven taxonomy and requirements. In: Faliszewski P, Mascardi V, Pelachaud C, et~al (eds) 21st International Conference on Autonomous Agents and Multiagent Systems, {AAMAS} 2022, Auckland, New Zealand, May 9-13, 2022. International Foundation for Autonomous Agents and Multiagent Systems {(IFAAMAS)}, pp 172--180, \doi{10.5555/3535850.3535871}, \urlprefix\url{https://www.ifaamas.org/Proceedings/aamas2022/pdfs/p172.pdf}

\bibitem[{Carelli et~al(2025)Carelli, Finkbeiner, and Siber}]{CarelliFS25}
Carelli M, Finkbeiner B, Siber J (2025) Closure and complexity of temporal causality. In: 40th Annual {ACM/IEEE} Symposium on Logic in Computer Science, {LICS} 2025, Singapore, June 23-26, 2025. {IEEE}

\bibitem[{Carnielli and Coniglio(2020)}]{sep-logic-combining}
Carnielli W, Coniglio ME (2020) {Combining Logics}. In: Zalta EN (ed) The {Stanford} Encyclopedia of Philosophy, {F}all 2020 edn. Metaphysics Research Lab, Stanford University

\bibitem[{Clarkson and Schneider(2010)}]{ClarksonS10}
Clarkson MR, Schneider FB (2010) Hyperproperties. J Comput Secur 18(6):1157--1210. \urlprefix\url{https://doi.org/10.3233/JCS-2009-0393}

\bibitem[{Coenen et~al(2019)Coenen, Finkbeiner, Hahn, and Hofmann}]{CoenenFHH19}
Coenen N, Finkbeiner B, Hahn C, et~al (2019) The hierarchy of hyperlogics. In: 34th Annual {ACM/IEEE} Symposium on Logic in Computer Science, {LICS} 2019, Vancouver, BC, Canada, June 24-27, 2019. {IEEE}, pp 1--13, \urlprefix\url{https://doi.org/10.1109/LICS.2019.8785713}

\bibitem[{Coenen et~al(2022{\natexlab{a}})Coenen, Dachselt, Finkbeiner, Frenkel, Hahn, Horak, Metzger, and Siber}]{CoenenDFFHHMS22}
Coenen N, Dachselt R, Finkbeiner B, et~al (2022{\natexlab{a}}) Explaining hyperproperty violations. In: Shoham S, Vizel Y (eds) Computer Aided Verification - 34th International Conference, {CAV} 2022, Haifa, Israel, August 7-10, 2022, Proceedings, Part {I}, Lecture Notes in Computer Science, vol 13371. Springer, pp 407--429, \doi{10.1007/978-3-031-13185-1\_20}, \urlprefix\url{https://doi.org/10.1007/978-3-031-13185-1\_20}

\bibitem[{Coenen et~al(2022{\natexlab{b}})Coenen, Finkbeiner, Frenkel, Hahn, Metzger, and Siber}]{CoenenFFHMS22}
Coenen N, Finkbeiner B, Frenkel H, et~al (2022{\natexlab{b}}) Temporal causality in reactive systems. In: Bouajjani A, Hol{\'{\i}}k L, Wu Z (eds) Automated Technology for Verification and Analysis - 20th International Symposium, {ATVA} 2022, Virtual Event, October 25-28, 2022, Proceedings, Lecture Notes in Computer Science, vol 13505. Springer, pp 208--224, \doi{10.1007/978-3-031-19992-9\_13}, \urlprefix\url{https://doi.org/10.1007/978-3-031-19992-9\_13}

\bibitem[{Darwiche and Ji(2022)}]{DarwicheJ22}
Darwiche A, Ji C (2022) On the computation of necessary and sufficient explanations. In: Thirty-Sixth {AAAI} Conference on Artificial Intelligence, {AAAI} 2022, Thirty-Fourth Conference on Innovative Applications of Artificial Intelligence, {IAAI} 2022, The Twelveth Symposium on Educational Advances in Artificial Intelligence, {EAAI} 2022 Virtual Event, February 22 - March 1, 2022. {AAAI} Press, pp 5582--5591, \urlprefix\url{https://ojs.aaai.org/index.php/AAAI/article/view/20498}

\bibitem[{Fagin et~al(1995)Fagin, Halpern, Moses, and Vardi}]{FaginHMV1995}
Fagin R, Halpern JY, Moses Y, et~al (1995) Reasoning About Knowledge. {MIT} Press, \doi{10.7551/mitpress/5803.001.0001}, \urlprefix\url{https://doi.org/10.7551/mitpress/5803.001.0001}

\bibitem[{Finkbeiner and Siber(2023)}]{FinkbeinerS23}
Finkbeiner B, Siber J (2023) Counterfactuals modulo temporal logics. In: Piskac R, Voronkov A (eds) {LPAR} 2023: Proceedings of 24th International Conference on Logic for Programming, Artificial Intelligence and Reasoning, Manizales, Colombia, 4-9th June 2023, EPiC Series in Computing, vol~94. EasyChair, pp 181--204, \doi{10.29007/qtw7}, \urlprefix\url{https://doi.org/10.29007/qtw7}

\bibitem[{Finkbeiner and Zimmermann(2017)}]{FinkbeinerZ17}
Finkbeiner B, Zimmermann M (2017) The first-order logic of hyperproperties. In: Vollmer H, Vall{\'{e}}e B (eds) 34th Symposium on Theoretical Aspects of Computer Science, {STACS} 2017, March 8-11, 2017, Hannover, Germany, LIPIcs, vol~66. Schloss Dagstuhl - Leibniz-Zentrum f{\"{u}}r Informatik, pp 30:1--30:14, \doi{10.4230/LIPIcs.STACS.2017.30}, \urlprefix\url{https://doi.org/10.4230/LIPIcs.STACS.2017.30}

\bibitem[{Finkbeiner et~al(2024)Finkbeiner, Frenkel, Metzger, and Siber}]{FinkbeinerFMS24}
Finkbeiner B, Frenkel H, Metzger N, et~al (2024) Synthesis of temporal causality. In: Gurfinkel A, Ganesh V (eds) Computer Aided Verification - 36th International Conference, {CAV} 2024, Montreal, QC, Canada, July 24-27, 2024, Proceedings, Part {III}, Lecture Notes in Computer Science, vol 14683. Springer, pp 87--111, \doi{10.1007/978-3-031-65633-0\_5}, \urlprefix\url{https://doi.org/10.1007/978-3-031-65633-0\_5}

\bibitem[{Finkbeiner et~al(2025)Finkbeiner, Frenkel, and Siber}]{FinkbeinerFS25}
Finkbeiner B, Frenkel H, Siber J (2025) An information-flow perspective on explainability requirements: Specification and verification. In: 22nd International Conference on Principles of Knowledge Representation and Reasoning, {KR} 2025, Melbourne, Australia, November 11-17, 2025, (to appear)

\bibitem[{Goyal et~al(2019)Goyal, Wu, Ernst, Batra, Parikh, and Lee}]{GoyalWEBPL19}
Goyal Y, Wu Z, Ernst J, et~al (2019) Counterfactual visual explanations. In: Chaudhuri K, Salakhutdinov R (eds) Proceedings of the 36th International Conference on Machine Learning, {ICML} 2019, 9-15 June 2019, Long Beach, California, {USA}, Proceedings of Machine Learning Research, vol~97. {PMLR}, pp 2376--2384, \urlprefix\url{http://proceedings.mlr.press/v97/goyal19a.html}

\bibitem[{Halpern(1999)}]{Halpern99}
Halpern JY (1999) Hypothetical knowledge and counterfactual reasoning. Int J Game Theory 28(3):315--330. \doi{10.1007/s001820050113}, \urlprefix\url{https://doi.org/10.1007/s001820050113}

\bibitem[{Halpern and O'Neill(2008)}]{HalpernO08}
Halpern JY, O'Neill KR (2008) Secrecy in multiagent systems. {ACM} Trans Inf Syst Secur 12(1):5:1--5:47. \doi{10.1145/1410234.1410239}, \urlprefix\url{https://doi.org/10.1145/1410234.1410239}

\bibitem[{Halpern and Pearl(2005{\natexlab{a}})}]{HalpernP05a}
Halpern JY, Pearl J (2005{\natexlab{a}}) Causes and explanations: A structural-model approach. part i: Causes. The British Journal for the Philosophy of Science 56(4):843--887. \urlprefix\url{http://www.jstor.org/stable/3541870}

\bibitem[{Halpern and Pearl(2005{\natexlab{b}})}]{HalpernP05b}
Halpern JY, Pearl J (2005{\natexlab{b}}) Causes and explanations: A structural-model approach. part ii: Explanations. The British Journal for the Philosophy of Science 56(4):889--911. \urlprefix\url{http://www.jstor.org/stable/3541871}

\bibitem[{Halpern et~al(2004)Halpern, van~der Meyden, and Vardi}]{HalpernMV04}
Halpern JY, van~der Meyden R, Vardi MY (2004) Complete axiomatizations for reasoning about knowledge and time. {SIAM} J Comput 33(3):674--703. \doi{10.1137/S0097539797320906}, \urlprefix\url{https://doi.org/10.1137/S0097539797320906}

\bibitem[{Hofmann(2022)}]{Hofmann22}
Hofmann J (2022) Logical methods for the hierarchy of hyperlogics. PhD thesis, Saarland University, Saarbr{\"{u}}cken, Germany, \urlprefix\url{https://publikationen.sulb.uni-saarland.de/handle/20.500.11880/35154}

\bibitem[{Horak et~al(2022)Horak, Coenen, Metzger, Hahn, Flemisch, M{\'{e}}ndez, Dimov, Finkbeiner, and Dachselt}]{HorakCMHFMDFD22}
Horak T, Coenen N, Metzger N, et~al (2022) Visual analysis of hyperproperties for understanding model checking results. {IEEE} Trans Vis Comput Graph 28(1):357--367. \doi{10.1109/TVCG.2021.3114866}, \urlprefix\url{https://doi.org/10.1109/TVCG.2021.3114866}

\bibitem[{Kamp(1968)}]{Kamp68}
Kamp JAW (1968) Tense logic and the theory of linear order. University of California, Los Angeles

\bibitem[{Khan and Lesp{\'{e}}rance(2021)}]{KhanL21}
Khan SM, Lesp{\'{e}}rance Y (2021) Knowing why - on the dynamics of knowledge about actual causes in the situation calculus. In: Dignum F, Lomuscio A, Endriss U, et~al (eds) {AAMAS} '21: 20th International Conference on Autonomous Agents and Multiagent Systems, Virtual Event, United Kingdom, May 3-7, 2021. {ACM}, pp 701--709, \doi{10.5555/3463952.3464037}, \urlprefix\url{https://www.ifaamas.org/Proceedings/aamas2021/pdfs/p701.pdf}

\bibitem[{Khan and Rostamigiv(2023)}]{DBLP:conf/ecai/KhanR23}
Khan SM, Rostamigiv M (2023) On explaining agent behaviour via root cause analysis: {A} formal account grounded in theory of mind. In: Gal K, Now{\'{e}} A, Nalepa GJ, et~al (eds) {ECAI} 2023 - 26th European Conference on Artificial Intelligence, September 30 - October 4, 2023, Krak{\'{o}}w, Poland - Including 12th Conference on Prestigious Applications of Intelligent Systems {(PAIS} 2023), Frontiers in Artificial Intelligence and Applications, vol 372. {IOS} Press, pp 1239--1247, \doi{10.3233/FAIA230401}, \urlprefix\url{https://doi.org/10.3233/FAIA230401}

\bibitem[{K{\"{o}}hl et~al(2019)K{\"{o}}hl, Baum, Langer, Oster, Speith, and Bohlender}]{KohlBLOSB19}
K{\"{o}}hl MA, Baum K, Langer M, et~al (2019) Explainability as a non-functional requirement. In: Damian DE, Perini A, Lee S (eds) 27th {IEEE} International Requirements Engineering Conference, {RE} 2019, Jeju Island, Korea (South), September 23-27, 2019. {IEEE}, pp 363--368, \doi{10.1109/RE.2019.00046}, \urlprefix\url{https://doi.org/10.1109/RE.2019.00046}

\bibitem[{Ladner and Reif(1986)}]{LadnerR86}
Ladner RE, Reif JH (1986) The logic of distributed protocols. In: Halpern JY (ed) Proceedings of the 1st Conference on Theoretical Aspects of Reasoning about Knowledge, Monterey, CA, USA, March 1986. Morgan Kaufmann, pp 207--222

\bibitem[{Langer et~al(2021)Langer, Oster, Speith, Hermanns, K{\"{a}}stner, Schmidt, Sesing, and Baum}]{LangerOSHKSSB21}
Langer M, Oster D, Speith T, et~al (2021) What do we want from explainable artificial intelligence (xai)? - {A} stakeholder perspective on {XAI} and a conceptual model guiding interdisciplinary {XAI} research. Artif Intell 296:103473. \doi{10.1016/j.artint.2021.103473}, \urlprefix\url{https://doi.org/10.1016/j.artint.2021.103473}

\bibitem[{Lewis(1973)}]{Lewis73}
Lewis DK (1973) Counterfactuals. Cambridge, MA, USA: Blackwell

\bibitem[{Lewis(1986)}]{Lewis86}
Lewis DK (1986) Causal explanation. In: Lewis D (ed) Philosophical Papers Vol. II. Oxford University Press, p 214--240

\bibitem[{Lichtenstein et~al(1985)Lichtenstein, Pnueli, and Zuck}]{LichtensteinPZ85}
Lichtenstein O, Pnueli A, Zuck LD (1985) The glory of the past. In: Parikh R (ed) Logics of Programs, Conference, Brooklyn College, New York, NY, USA, June 17-19, 1985, Proceedings, Lecture Notes in Computer Science, vol 193. Springer, pp 196--218, \doi{10.1007/3-540-15648-8\_16}, \urlprefix\url{https://doi.org/10.1007/3-540-15648-8\_16}

\bibitem[{Liu and Lorini(2023)}]{LiuL23}
Liu X, Lorini E (2023) A unified logical framework for explanations in classifier systems. J Log Comput 33(2):485--515. \doi{10.1093/logcom/exac102}, \urlprefix\url{https://doi.org/10.1093/logcom/exac102}

\bibitem[{van~der Meyden and Shilov(1999)}]{MeydenS99}
van~der Meyden R, Shilov NV (1999) Model checking knowledge and time in systems with perfect recall (extended abstract). In: Rangan CP, Raman V, Ramanujam R (eds) Foundations of Software Technology and Theoretical Computer Science, 19th Conference, Chennai, India, December 13-15, 1999, Proceedings, Lecture Notes in Computer Science, vol 1738. Springer, pp 432--445, \doi{10.1007/3-540-46691-6\_35}, \urlprefix\url{https://doi.org/10.1007/3-540-46691-6\_35}

\bibitem[{van~der Meyden and Su(2004)}]{MeydenS04}
van~der Meyden R, Su K (2004) Symbolic model checking the knowledge of the dining cryptographers. In: 17th {IEEE} Computer Security Foundations Workshop, {(CSFW-17} 2004), 28-30 June 2004, Pacific Grove, CA, {USA}. {IEEE} Computer Society, p 280, \doi{10.1109/CSFW.2004.19}, \urlprefix\url{https://doi.ieeecomputersociety.org/10.1109/CSFW.2004.19}

\bibitem[{van~der Meyden and Vardi(1998)}]{MeydenV98}
van~der Meyden R, Vardi MY (1998) Synthesis from knowledge-based specifications (extended abstract). In: Sangiorgi D, de~Simone R (eds) {CONCUR} '98: Concurrency Theory, 9th International Conference, Nice, France, September 8-11, 1998, Proceedings, Lecture Notes in Computer Science, vol 1466. Springer, pp 34--49, \doi{10.1007/BFb0055614}, \urlprefix\url{https://doi.org/10.1007/BFb0055614}

\bibitem[{Miller(2019)}]{Miller19}
Miller T (2019) Explanation in artificial intelligence: Insights from the social sciences. Artif Intell 267:1--38. \doi{10.1016/j.artint.2018.07.007}, \urlprefix\url{https://doi.org/10.1016/j.artint.2018.07.007}

\bibitem[{Pnueli(1977)}]{Pnueli77}
Pnueli A (1977) The temporal logic of programs. In: 18th Annual Symposium on Foundations of Computer Science, Providence, Rhode Island, USA, 31 October - 1 November 1977. {IEEE} Computer Society, pp 46--57, \doi{10.1109/SFCS.1977.32}, \urlprefix\url{https://doi.org/10.1109/SFCS.1977.32}

\bibitem[{Poyiadzi et~al(2020)Poyiadzi, Sokol, Santos{-}Rodr{\'{\i}}guez, Bie, and Flach}]{PoyiadziSSBF20}
Poyiadzi R, Sokol K, Santos{-}Rodr{\'{\i}}guez R, et~al (2020) {FACE:} feasible and actionable counterfactual explanations. In: Markham AN, Powles J, Walsh T, et~al (eds) {AIES} '20: {AAAI/ACM} Conference on AI, Ethics, and Society, New York, NY, USA, February 7-8, 2020. {ACM}, pp 344--350, \doi{10.1145/3375627.3375850}, \urlprefix\url{https://doi.org/10.1145/3375627.3375850}

\bibitem[{Rabe(2016)}]{Rabe16}
Rabe MN (2016) A temporal logic approach to information-flow control. PhD thesis, Saarland University, \urlprefix\url{http://scidok.sulb.uni-saarland.de/volltexte/2016/6387/}

\bibitem[{Rosenfeld and Richardson(2019)}]{RosenfeldR19}
Rosenfeld A, Richardson A (2019) Explainability in human-agent systems. Auton Agents Multi Agent Syst 33(6):673--705. \doi{10.1007/s10458-019-09408-y}, \urlprefix\url{https://doi.org/10.1007/s10458-019-09408-y}

\bibitem[{Sakama(2021)}]{Sakama21}
Sakama C (2021) Deception in epistemic causal logic. In: Sarkadi S, Wright B, Masters P, et~al (eds) Deceptive AI. Springer International Publishing, Cham, pp 105--123

\bibitem[{Sandu(2021)}]{Sandu21}
Sandu AS (2021) Knowledge of counterfactuals. PhD thesis, Cornell University

\bibitem[{Schlicker et~al(2021)Schlicker, Langer, {\"{O}}tting, Baum, K{\"{o}}nig, and Wallach}]{SchlickerLOBKW21}
Schlicker N, Langer M, {\"{O}}tting SK, et~al (2021) What to expect from opening up 'black boxes'? comparing perceptions of justice between human and automated agents. Comput Hum Behav 122:106837. \doi{10.1016/j.chb.2021.106837}, \urlprefix\url{https://doi.org/10.1016/j.chb.2021.106837}

\bibitem[{Stalnaker(1981)}]{Stalnaker81}
Stalnaker R (1981) A Theory of Conditionals, Springer Netherlands, Dordrecht, pp 41--55. \doi{10.1007/978-94-009-9117-0_2}, \urlprefix\url{https://doi.org/10.1007/978-94-009-9117-0_2}

\bibitem[{Stalnaker(2006)}]{Stalnaker06}
Stalnaker R (2006) On logics of knowledge and belief. Philosophical Studies: An International Journal for Philosophy in the Analytic Tradition 128(1):169--199. \urlprefix\url{http://www.jstor.org/stable/4321718}

\bibitem[{Wachter et~al(2018)Wachter, Mittelstadt, and Russell}]{Wachter18}
Wachter S, Mittelstadt B, Russell C (2018) Counterfactual explanations without opening the black box: automated decisions and the gdpr. Harvard Journal of Law and Technology 31(2):841--887

\bibitem[{Ziemek et~al(2022)Ziemek, Piribauer, Funke, Jantsch, and Baier}]{ZiemekPFJB22}
Ziemek R, Piribauer J, Funke F, et~al (2022) Probabilistic causes in markov chains. Innov Syst Softw Eng 18(3):347--367. \doi{10.1007/s11334-022-00452-8}, \urlprefix\url{https://doi.org/10.1007/s11334-022-00452-8}

\end{thebibliography}

\end{document}